\documentclass[12pt]{article}
\usepackage[utf8]{inputenc}
\usepackage{fullpage}
\usepackage{authblk}

\usepackage{verbatim}
\usepackage[T1]{fontenc}
\usepackage[tt=false,type1=true]{libertine}

\usepackage{hyperref}

\usepackage{amsthm}

\usepackage{complexity}
\usepackage{booktabs}
\let\R\relax
\newcommand{\R}{\mathbb{R}}

\usepackage{amsmath}
\usepackage[libertine]{newtxmath}
\usepackage{mathtools}
\usepackage{thm-restate}
\usepackage{bm}
\usepackage{bbm}

\usepackage{authblk}

\usepackage{xcolor}

\usepackage{tikz}
\usetikzlibrary{shapes,arrows,positioning,fit,backgrounds,decorations.pathreplacing,calc}

\usepackage{pgfplots}
\pgfplotsset{compat=1.18}

\renewcommand{\vec}[1]{\bm{#1}}

\newcommand{\N}{\mathbb{N}}

\newcommand{\vx}{\vec{x}}
\newcommand{\vm}{\vec{m}}

\DeclareMathOperator{\FIXP}{\mathsf{FIXP}}

\newcommand{\Tgood}[1][]{\mathcal{T}_{\mathsf{sound}#1}}

\newcommand{\gaplabel}{\textsc{GapLabelCover}}

\newcommand{\revenue}{\mathsf{Revenue}}
\newcommand{\welfare}{\mathsf{LWel}}
\newcommand{\poa}{\mathsf{PoA}}
\newcommand{\opt}{\mathsf{OPT}}
\newcommand{\signal}{\mathsf{sig}}

\newcommand{\msafe}[1][]{m_{\mathsf{safe} #1}}

\newcommand{\defeq}{\coloneqq}

\newclass{\polyAPX}{poly-APX}

\usepackage[american]{circuitikz}

\theoremstyle{plain}
\newtheorem{theorem}{Theorem}[section]
\newtheorem{proposition}[theorem]{Proposition}
\newtheorem{lemma}[theorem]{Lemma}
\newtheorem{claim}[theorem]{Claim}
\newtheorem{corollary}[theorem]{Corollary}
\newtheorem{conjecture}[theorem]{Conjecture}
\theoremstyle{definition}
\newtheorem{definition}[theorem]{Definition}

\theoremstyle{remark}
\newtheorem{remark}[theorem]{Remark}



\newcommand{\declarecolor}[2]{\definecolor{#1}{RGB}{#2}\expandafter\newcommand\csname #1\endcsname[1]{\textcolor{#1}{##1}}}
\declarecolor{White}{255, 255, 255}
\declarecolor{Black}{0, 0, 0}
\declarecolor{Maroon}{128, 0, 0}
\declarecolor{Coral}{255, 127, 80}
\declarecolor{Red}{182, 21, 21}
\declarecolor{LimeGreen}{50, 205, 50}
\declarecolor{DarkGreen}{0, 80, 0}
\declarecolor{Purple}{146, 42, 158}
\declarecolor{Navy}{0, 0, 128}
\declarecolor{LightBlue}{84, 101, 202}
\definecolor{mydarkblue}{rgb}{0,0.08,0.45}

\hypersetup{ %
    pdftitle={},
    pdfkeywords={},
    pdfborder=0 0 0,
    pdfpagemode=UseNone,
    colorlinks=true,
    linkcolor=Maroon,
    citecolor=Navy,
    filecolor=Black,
    urlcolor=Black,
}

\usepackage{natbib}
\usepackage{cleveref}

\crefname{conjecture}{Conjecture}{Conjectures}
\crefname{claim}{Claim}{Claims}

\title{Tight Inapproximability for Welfare-Maximizing Autobidding Equilibria}

\author[1]{Ioannis Anagnostides\thanks{Part of this work was performed while at Google DeepMind.}}
\author[2]{Ian Gemp}
\author[2,3]{Georgios Piliouras}
\author[4]{Kelly Spendlove}

\affil[1]{Carnegie Mellon University}
\affil[2]{Google DeepMind}
\affil[3]{Singapore University of Technology and Design}
\affil[4]{Google}
\affil[ ]{\small \texttt{ianagnos@cs.cmu.edu}, \texttt{imgemp@google.com}, \texttt{gpil@google.com},
\texttt{spendlove@google.com}}

\date{}

\begin{document}

\maketitle
\thispagestyle{empty}

\begin{abstract}
     We examine the complexity of computing welfare- and revenue-maximizing equilibria in autobidding second-price auctions subject to \emph{return-on-spend (RoS)} constraints. We show that computing an autobidding equilibrium that approximates the welfare-optimal one within a factor of $2 - \epsilon$ is $\mathsf{NP}$-hard for any constant $\epsilon > 0$. Moreover, deciding whether there exists an autobidding equilibrium that attains a $1/2 + \epsilon$ fraction of the optimal welfare---unfettered by equilibrium constraints--- is $\mathsf{NP}$-hard for any constant $\epsilon > 0$. This hardness result is tight in view of the fact that the price of anarchy (PoA) is at most $2$, and shows that deciding whether a non-trivial autobidding equilibrium exists---one that is even marginally better than the worst-case guarantee---is intractable. For revenue, we establish a stronger logarithmic inapproximability, while under the projection games conjecture, our reduction rules out even a polynomial approximation factor. These results significantly strengthen the $\mathsf{APX}$-hardness of Li and Tang (AAAI '24). Furthermore, we refine our reduction in the presence of ML advice concerning the buyers' valuations, revealing again a close connection between the inapproximability threshold and PoA bounds. Finally, we examine relaxed notions of equilibrium attained by simple learning algorithms, establishing constant inapproximability for both revenue and welfare.
\end{abstract}

\clearpage
\thispagestyle{empty}
\tableofcontents
\clearpage
\setcounter{page}{1}

\section{Introduction}

The proliferation of autobidding systems in recent years has fundamentally transformed online advertising~\citep{Aggarwal24:Autobidding}. Rather than specifying bids on a granular, per-keyword basis, advertisers can now submit high-level objectives---such as maximizing conversions subject to \emph{return-on-spend (RoS)} or budget constraints---to automated agents that bid on their behalf to optimally meet those objectives. This new paradigm drastically simplifies the interface for the advertisers, while simultaneously enabling more effective real-time decision-making.

While autobidding systems confer tangible benefits, their widespread adoption raises pressing questions concerning market efficiency: when multiple autobidders interact under a particular auction format, what can we expect in terms of the system's behavior? A foundational result by~\citet{Aggarwal19:Autobidding} showed that the \emph{price of anarchy}---a measure of equilibrium suboptimality---with respect to \emph{liquid welfare} under parallel second-price auctions is at most $2$; that is, even the worst-case equilibrium delivers at least half of the optimal liquid welfare. However, when billions of dollars are at stake---global advertising expenditure topped \$1 trillion for the first time in 2024~\citep{WARC24}---an approximation factor of $2$ is hardly reassuring.

In this paper, we examine the computational complexity of finding welfare- and revenue-maximizing autobidding equilibria. The computational lens is especially relevant here in light of the scale and dynamic nature of modern advertising platforms. Barriers to computing high-quality equilibria fundamentally constrain the effectiveness of general-purpose intervention mechanisms. Conversely, the existence of near-optimal approximation algorithms would imply that a system designer could, in principle, steer the autobidders toward high-quality equilibria.

The price of anarchy bound of~\citet{Aggarwal19:Autobidding} already provides a baseline: any autobidding equilibrium guarantees an approximation factor of $2$, even relative to the unconstrained welfare-optimal allocation. Furthermore, \citet{Li24:Vulnerabilities} recently showed that maximizing welfare or revenue over the set of autobidding equilibria is \APX-hard, implying that---subject to $\P \neq \NP$---the problem cannot be approximated within \emph{some} constant factor (that could be very close to 1). This leaves a substantial gap between the known hardness threshold and the baseline guarantee of 2. Our main contribution is to close this gap, establishing tight inapproximability for computing welfare- and revenue-maximizing autobidding equilibria.

Before we present our results, we provide the necessary background on autobidding systems.

\paragraph{Background} We consider a setting comprising $n$ autobidders and an auctioneer allocating $k$ items. Each autobidder $i$ submits a bid $b_{i j}$ for item $j \in [k]$. We denote the values of autobidder $i \in [n]$ by $\{ v_{ij} \}_{j \in [k]}$, which are additive. We assume that autobidders apply \emph{uniform bid scaling}, meaning the bid profile of each bidder is a scalar multiple of their valuation profile: $b_{i j} = m_i \cdot v_{i j}$, for some multiplier $m_i > 0$. As a result, each autobidder optimizes only a single parameter. This restriction is well-motivated: for auction formats that are truthful in the quasi-linear model, uniform bid scaling is known to be optimal~\citep{Aggarwal19:Autobidding}.

The auctioneer takes as input the bids and produces a potentially fractional allocation $\vec{x} \in \R_{\geq 0}^{n \times k}$ satisfying $\sum_{i=1}^n x_{ij} = 1$ for each item $j \in [k]$, along with payments $\vec{p} \in \R^{n \times k}_{\geq 0}$. Throughout this paper, we consider a (parallel) second-price allocation rule, whereby each item is allotted to the highest bidder at a price equal to the second-highest bid; in the case of ties, careful handling is required to guarantee existence of an equilibrium~\citep{Li24:Vulnerabilities}.

\paragraph{RoS constraint} The advertiser (or buyer) provides their autobidder with a \emph{return-on-spend (RoS)} target, denoted by $\tau_i \geq 1$. The goal of each autobidder $i$ is to select bids that maximize the total value $\sum_{j = 1}^k x_{ij} v_{ij}$ subject to the \emph{RoS constraint} parameterized by $\tau_i$:
\begin{equation}
    \label{eq:RoS}
    \max_{\vec{b}_i} \sum_{j=1}^k v_{i j} x_{i j}(\vec{b}) \quad \text{subject to} \quad \tau_i \cdot \sum_{j=1}^k p_{i j}(\vec{b}) \leq \sum_{j=1}^k v_{i j} x_{i j}(\vec{b}).
\end{equation}
By rescaling the values, we can assume without loss of generality that $\tau_i = 1$ for each autobidder $i$. A compelling aspect of our lower bounds is that they hold even in the presence of solely RoS constraints, without any additional budget constraints. In this setting, \emph{liquid welfare} coincides with the usual notion of welfare (\Cref{sec:prels} contains the precise definitions).

An \emph{autobidding equilibrium} (\Cref{def:autobidequil}, introduced by~\citealp{Li24:Vulnerabilities}) describes a stable state where multiple autobidders simultaneously maximize value subject to individual RoS constraints. Throughout this paper, when we use the term ``autobidding equilibrium,'' it is with respect to a second-price auction format. We caution that while we adopt the terminology of~\citet{Li24:Vulnerabilities}, this concept appears under different names in the literature (\emph{e.g.},~\citealp{Filos24:Ppad}).

\subsection{Our results}

We begin with the (liquid) welfare objective, where we establish a tight inapproximability result. 

\begin{theorem}
    \label{theorem:informal-welfare}
For any constant $\epsilon > 0$, it is \NP-hard to compute an autobidding equilibrium that approximates the welfare-optimal one within a factor $2 - \epsilon$.
\end{theorem}

This matches the PoA bound of 2, establishing that no non-trivial approximation algorithm exists for this problem.\footnote{Computing a $2$-approximation is in $\FIXP$ by virtue of the membership result of~\citet{Filos24:Ppad}. A problem in $\FIXP$ cannot be \NP-hard unless $\coNP \subseteq \NP_{\mathbb{R}}$, where $\NP_{\mathbb{R}}$ is the class of problems polynomial-time reducible to deciding the existential theory of the reals~\citep{Etessami07:Complexity}. For context, if a problem in $\PPAD$ is \NP-hard then $\NP = \coNP$~\citep{Johnson88:Easy}.} 
In other words, while any autobidding equilibrium guarantees a $2$-approximation on account of the price of anarchy bound of~\citet{Aggarwal19:Autobidding}, we cannot hope to improve upon this factor. For a decision version of this question, we show the following.

\begin{theorem}
    For any constant $\epsilon > 0$, deciding whether there exists an autobidding equilibrium that attains at least a $1/2 + \epsilon$ fraction of the (unconstrained) optimal welfare is \NP-hard.
\end{theorem}

This result is tight: the problem becomes trivial when $\epsilon = 0$ since any autobidding equilibrium delivers at least half of the welfare-optimal allocation. As a result, we have shown that deciding whether there is an autobidding equilibrium that is even marginally better than the worst one is \NP-hard. Adopting the terminology of~\citet{Barman15:Finding}, we can summarize this finding as follows.

\begin{corollary}[Computing non-trivial autobidding equilibria is \NP-hard]
    \label{cor:nontrivial}
    For any constant $\epsilon > 0$, computing an autobidding equilibrium whose welfare is at least a $(1+\epsilon)$ fraction higher than that with the lowest welfare (that is, a non-trivial autobidding equilibrium) is \NP-hard.
\end{corollary}
To place this into context, \citet{Barman15:Finding} established hardness for computing non-trivial (coarse) correlated equilibria in succinct multi-player games. \Cref{cor:nontrivial} provides an analogous hardness result for autobidding equilibria.

Our results settle the approximability of welfare maximization in (second-price) autobidding auctions, and significantly strengthen the previous \APX-hardness established by~\citet{Li24:Vulnerabilities}. We further note that while computing an autobidding equilibrium is \PPAD-hard in general~\citep{Li24:Vulnerabilities}, our reduction is based on a class of instances in which finding an autobidding equilibrium is trivial; it is the \emph{equilibrium selection} that drives our hardness result.

For revenue maximization, we establish a stronger inapproximability result, ruling out attaining even a (sub)logarithmic approximation factor.

\begin{theorem}
    \label{theorem:informal-revenue}
Computing an autobidding equilibrium that approximates the revenue-optimal one within some factor $\log^{\Omega(1)}(nk)$ is \NP-hard.
\end{theorem}

This stark contrast reveals a fundamental difference between liquid welfare and revenue in autobidding systems. Once again, \Cref{theorem:informal-revenue} significantly strengthens the \APX-hardness of~\citet{Li24:Vulnerabilities}. Furthermore, under the \emph{projection games conjecture} (\Cref{conj:projection})---a common hypothesis in complexity theory---we show that our reduction yields \emph{polynomial inapproximability} for revenue maximization (\Cref{cor:revenue-projectiongames}).

\paragraph{Refinement under ML advice} Next, we refine our hardness results in the presence of ML advice concerning the underlying valuation profile. Specifically, we adopt the model of~\citet{Balseiro21:Robust}, which posits access to a signal that provides a multiplicative approximation (\Cref{def:approx-signal}). This signal is then used as a reserve price. \citet{Balseiro21:Robust} established tight PoA bounds parameterized by the accuracy of the signals, denoted by $\gamma > 0$; when $\gamma = 1$ the auctioneer has perfect information concerning valuations, while, at the other end of the spectrum, when $\gamma \approx 0$, we revert to the standard model.

We identify again a close connection between the PoA bounds and the inapproximability threshold. Specifically, we use our reduction to establish the following hardness results.

\begin{theorem}
    \label{theorem:informal-refine-welfare}
    Under $\gamma$-approximate signals with any constant $\gamma \in [0, 1)$ and any constant $\epsilon > 0$, computing an autobidding equilibrium that approximates the welfare-optimal one within a factor $2/(1 + \gamma) - \epsilon$ is \NP-hard.
\end{theorem}
This recovers~\Cref{theorem:informal-welfare} as a special case when $\gamma =0$, but goes much further by parameterizing the inapproximability threshold in terms of $\gamma$. For revenue, we refine~\Cref{theorem:informal-revenue} as follows.

\begin{theorem}
    \label{theorem:informal-refine-revenue}
    Under $\gamma$-approximate signals with any constant $\gamma \in [0,1 )$ and any constant $\epsilon > 0$, computing an autobidding equilibrium that approximates the revenue-optimal one within a factor $1/(\gamma + \epsilon)$ is \NP-hard.
\end{theorem}

\paragraph{Inapproximability for learning dynamics} Our foregoing results concern autobidding equilibria. However, the learning dynamics typically employed in practice are not guaranteed to converge to autobidding equilibria~\citep{Gaitonde23:Budget,Lucier24:Autobidders,Leme24:Complex}. As a result, a natural question is whether we can extend our analysis to a broader set of outcomes that can be attained by simple learning algorithms. Our primary goal here is to make the definition as broad as possible.

To this end, our first relaxation only imposes that the average RoS constraint is approximately satisfied. We show that revenue inapproximability persists under this very permissive definition.

\begin{theorem}
    \label{theorem:informal-revenue-learning}
    For any constant $\epsilon > 0$, computing a second-price sequence satisfying the time-average RoS constraints that approximates the revenue-optimal one within a factor $e/(e-1) - \epsilon$ is \NP-hard.
\end{theorem}

Turning to welfare, we need to make further assumptions concerning the autobidders' behavior; an autobidder always selecting a safe multiplier of $1$ is guaranteed to satisfy the RoS constraint in each round, but that clearly violates maximal pacing, which is a key prerequisite in the definition of autobidding equilibria. In this context, we make a mild assumption (\Cref{def:learning}) stipulating that whenever an autobidder has collected a large surplus over an extended period, its average multiplier must experience an increase---unless it was already maximal. We call the resulting notion a \emph{responsive learning sequence}.\footnote{We refrain from using the term ``equilibrium'' as the requirement we impose is quite weak.} 

\begin{theorem}
    \label{theorem:informal-welfare-learning}
    For any constant $\epsilon > 0$, computing a responsive learning sequence that approximates the welfare-optimal one within a factor $2e/(2e - 1) - \epsilon$ is \NP-hard.
\end{theorem}

For concreteness, we also provide a canonical update rule that produces a responsive learning sequence per~\Cref{def:learning} (\Cref{lemma:psi}).

\subsection{Technical approach}

A schematic summary of our reduction concerning autobidding equilibria is shown in~\Cref{fig:diagram}. Our starting point is the \emph{label cover} problem (\Cref{def:labelcover}), which is a mainstay in the line of work on the hardness of approximation. Here, we are given a bipartite graph $G$ whose nodes are to be labeled based on a finite alphabet $\Sigma$ to satisfy as many of the edge constraints as possible. Each edge constraint $e \in E$ is dictated by a given function $\Pi_e: \Sigma \to \Sigma$, prescribing a set of allowable pairs. The label cover problem can be written compactly as a constraint satisfaction problem over a set of Boolean variables $\{z_{u, \sigma} \}_{u \in V, \sigma \in \Sigma}$ in the form alluded to in~\Cref{fig:diagram}. What makes this problem particularly useful for our purposes is that it is \NP-hard to approximate even under the promise that we are given a gap-amplified instance where either all edge constraints can be satisfied or only a small fraction thereof can be satisfied (\Cref{theorem:labelcover}).

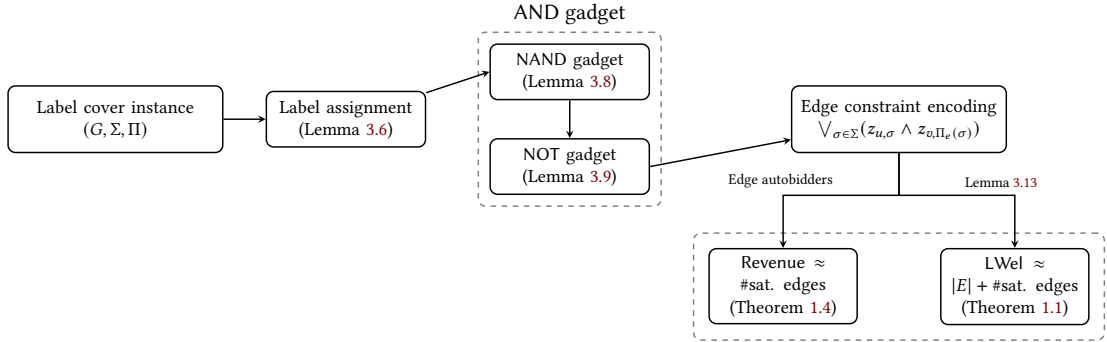
\begin{figure}[t]
    \centering
    \scalebox{0.7}{%
        \begin{tikzpicture}[
    node distance=1cm and 0.8cm,
    auto,
    block/.style={rectangle, draw, text width=6em, text centered, rounded corners, minimum height=3em, thick, fill=white, font=\footnotesize},
    gadget/.style={rectangle, draw, text width=6.5em, text centered, rounded corners, minimum height=2.5em, thick, fill=white, font=\footnotesize},
    instance/.style={rectangle, draw, text width=9em, text centered, rounded corners, minimum height=3em, thick, fill=white, font=\footnotesize},
    arrow/.style={->, >=stealth, thick},
    label/.style={font=\scriptsize}
]


\node[instance] (labelcover) {Label cover instance\\$(G, \Sigma, \Pi)$};

\node[gadget, right=of labelcover] (assignment) {Label assignment \\(\Cref{lemma:assignment})};

\node[gadget, right=1.2cm of assignment, yshift=0.9cm] (nand) {\textsf{NAND} gadget\\(\Cref{lemma:nand})};
\node[gadget, right=1.2cm of assignment, yshift=-0.9cm] (negation) {\textsf{NOT} gadget\\(\Cref{lemma:negation})};

\node[instance, right=1.2cm of assignment] (csp) at ($(nand)!0.5!(negation) + (3, 0)$) {Edge constraint encoding\\$ \bigvee_{\sigma \in \Sigma} (z_{u, \sigma} \land z_{v, \Pi_e(\sigma)})$};



\node[block, below=1.8cm of csp, xshift=-2.2cm] (revenue) {$\revenue \approx \#\text{sat. edges}$\\(\Cref{theorem:informal-revenue})};

\node[block, below=1.8cm of csp, xshift=2.2cm] (welfare) {$\welfare \approx |E| + \#\text{sat. edges}$\\(\Cref{theorem:informal-welfare})};

\begin{scope}[on background layer]
    \node[fit=(nand)(negation), 
          draw=gray, thick, dashed, rounded corners, inner sep=6pt] (andbox) {};
    
    \node[above] at (andbox.north) {\textsf{AND} gadget};
          
    \node[fit=(revenue)(welfare), 
          draw=gray, thick, dashed, rounded corners, inner sep=8pt,
          label={[gray, font=\bfseries\scriptsize, align=center]south:Inapproximability Results}] (results) {};
\end{scope}


\draw[arrow] (labelcover) -- (assignment);
\draw[arrow] (assignment) -- (nand.west);
\draw[arrow] (nand.south) -- (negation.north);
\draw[arrow] (negation.east) -- (csp);

\draw[arrow] (csp.south) -- ++(0,-0.8) -| node[label, near start, above left] {Edge autobidders} (revenue.north);
\draw[arrow] (csp.south) -- ++(0,-0.8) -| node[label, near start, above right] {\Cref{lemma:stealing}} (welfare.north);

\end{tikzpicture}%
    }
    \caption{Overview of our reduction for autobidding equilibria.}
    \label{fig:diagram}
\end{figure}

The name of the game now is to encode a label cover instance through an autobidding system. To do so, we associate a high multiplier $M \gg 1$ with a \textsf{True} value and a low multiplier $M \approx 1$ with a \textsf{False} value. We then develop a series of gadgets. First, the label assignment gadget (\Cref{lemma:assignment}) has the soundness property that at most one multiplier across a set of autobidders (representing the alphabet $\Sigma$) will be active, corresponding to the given label. Second, we implement an \textsf{AND} gadget for Boolean conjunction; this is done in two stages: a \textsf{NAND} gadget (\Cref{lemma:nand}) and a \textsf{NOT} gadget (\Cref{lemma:negation}). A crucial detail here is that the autobidding systems we design to implement these operations are low stakes in that the values/payments involved are negligible, so that the welfare or revenue remains unaffected. Finally, if each multiplier $m_{e, \sigma}$, for $e = (u, v)$, is associated with $z_{u, \sigma} \land z_{v, \Pi_e(\sigma)}$, we create an edge autobidder that competes with $\{ m_{e, \sigma} \}_{\sigma \in \Sigma}$, so that the price paid reflects $\bigvee_{\sigma \in \Sigma} (z_{u, \sigma} \land z_{v, \Pi_e(\sigma)})$.

With these ingredients at hand, we first create an autobidding system whose revenue is, in equilibrium, roughly at most the number of satisfied edges in the original label cover instance (soundness). Moreover, there exists an autobidding equilibrium whose revenue matches the number of satisfied edges (completeness). This paves the way to~\Cref{theorem:informal-revenue}. 

For welfare, we introduce an additional component, which guarantees that the edge autobidder will use its extra surplus to inefficiently capture part of another item (\Cref{lemma:stealing}). What is intriguing is that this basic interaction is the simplest example establishing the tightness of the PoA of 2; our reduction essentially embeds that for each edge constraint. The soundness property guarantees that the welfare of the overall system is, in equilibrium, roughly at most $|E|$ plus the number of satisfied edges, and the completeness property establishes that there exists an autobidding equilibrium that matches $|E|$ plus the number of satisfied edges, leading to~\Cref{theorem:informal-welfare}.

Our refinements under ML advice (\Cref{theorem:informal-refine-revenue,theorem:informal-refine-welfare}) follow this basic reduction, adapting the argument to account for the presence of reserve prices.

\paragraph{Learning dynamics} Turning to our hardness results concerning responsive learning sequences, our starting point is a class of CSPs that contains the maximum cover problem. This is a more benign class of problems than the one employed earlier, in that the gap is $1 - 1/e$ (\Cref{theorem:Feige}). However, here we do not need to implement an \textsf{AND} gadget, which we found to be particularly challenging unless one imposes strong assumptions on the behavior of the learning dynamics. We stress again that, for this set of results, our priority was to cover as broad a class as possible, at the cost of establishing a weaker intractability threshold. 

In this context, our reduction proceeds by showing that the label assignment guarantees soundness \emph{for most} rounds. This already suffices to arrive at~\Cref{theorem:informal-revenue-learning} concerning revenue under time-average RoS constraints by carefully setting the parameters of the reduction. For welfare, we use the additional property of the responsive learning sequence to argue that each clause autobidder will exhaust a large part of its surplus, which in turn means that it will inefficiently capture, on average, a large fraction of an item that it does not value highly.
\subsection{Further related work}

Before we proceed, we place our contributions in the context of the existing literature.

\paragraph{Equilibria in autobidding systems} \citet{Conitzer22:Multiplicative} introduced the notion of a \emph{second-price pacing equilibrium}, which targets settings with budget constraints. \citet{Chen24:Complexity} showed that it is $\PPAD$-hard to compute one (we also refer to~\citealp{Chen25:Constant} for an improved hardness result). A related notion is the~\emph{throttling equilibrium}~\citep{Chen21:Throttling}. The key reference point of our paper is~\citet{Li24:Vulnerabilities}, which formalized the notion of an autobidding equilibrium (\emph{cf.}~\citealp{Mehta22:Auction}); this concept specifically targets settings governed by RoS constraints, which is also the main focus of our paper.

\paragraph{Efficiency of equilibria and learning dynamics} Following the foundational work of~\citet{Aggarwal19:Autobidding}, it has been shown that randomization can improve the price of anarchy~\citep{Mehta22:Auction,Liaw23:Efficiency}. It is an open problem to design a randomized mechanism with price of anarchy strictly better than $2$ in general settings. The computational framework we introduce in this paper could shed light on this question: as observed by~\citet{Roughgarden14:Barriers}, further discussed below, intractability barriers circumscribe the efficiency one can hope for.

The intractability of computing equilibria has motivated the study of the efficiency of learning dynamics. \citet{Gaitonde23:Budget} show that the liquid welfare attained when all autobidders adopt certain gradient-based dynamics is at least half the optimal one. Crucially, this result does not hinge on the convergence of the dynamics. This result was subsequently strengthened by~\citet{Lucier24:Autobidders}. Moreover, \citet{Fikioris25:Liquid} provide liquid welfare guarantees in first-price auctions under some multiplicative approximation relative to the best fixed multiplier in hindsight. For a recent characterization of efficiency, we refer to~\citet{Baldeschi26:Optimal}.

\paragraph{Broader context} As we explained earlier, it has been shown that uniform bid scaling is in some sense optimal~\citep{Aggarwal19:Autobidding}. On the other hand, for auction formats that are not truthful, such as first-price auctions (FPAs) or generalized second-price auctions (GSPs), uniform bid scaling can produce suboptimal strategies~\citep{Deng21:Towards}. It is also worth noting that there has been some work exploring the setting where the RoS constraint is enforced separately for each item, rather than in aggregate. Notably, \citet{Wilkens17:GSP} showed that GSP is truthful in that setting.

For further pointers to the burgeoning body of work in autobidding systems, we refer to the recent survey of~\citet{Aggarwal24:Autobidding}.

\paragraph{Connection between hardness of approximation and PoA} An intriguing connection between price of anarchy and hardness of approximation was identified by~\citet{Roughgarden14:Barriers}. Specifically, \citet{Roughgarden14:Barriers} observed that one can harness a hardness of approximation result to arrive at a price of anarchy lower bound, conditional on natural complexity assumptions. In a similar vein, we find that in autobidding systems the price of anarchy bound precisely matches the inapproximability threshold of 2.

\paragraph{Computing non-trivial equilibria} In the context of succinct normal-form games, \citet{Barman15:Finding} showed that computing any non-trivial a coarse correlated equilibrium (CCE)---one whose welfare is even marginally better than that of the worst CCE---is \NP-hard. One of our results establishes a similar lower bound for autobidding equilibria (\Cref{cor:nontrivial}). That fact that computing the welfare-optimal (coarse) correlated equilibrium is \NP-hard was established in the seminal work of~\citet{Papadimitriou08:Computing}.
\section{Preliminaries}
\label{sec:prels}

Before we dive into our results, we provide additional background on autobidding systems and equilibria.

\paragraph{Welfare, revenue, and price of anarchy} 

We recall that $x_{i j}$ represents the fraction of item $j \in [k]$ allocated to autobidder $i \in [n]$; a fractional allocation can be thought of in probabilistic terms, with $x_{i j}$ corresponding to the probability that $i$ obtains item $j$. We let $\welfare(\vec{x}) \defeq \sum_{i=1}^n \min \{ B_i, \frac{1}{\tau_i} \sum_{j=1}^k x_{i j} v_{i j}  \} $ be the \emph{liquid welfare}, where $\tau_i \geq 1$ is the RoS parameter per~\eqref{eq:RoS} and $B_i > 0$ is the \emph{budget}. In what follows, we assume $B_i = +\infty$---that is, we lift the budget constraint---and $\tau_i = 1$ for all autobidders $i \in [n]$, so that liquid welfare reduces to the usual notion of welfare; an important aspect of our hardness results is that they persist without any budget constraints.\footnote{Assuming $\tau_i = 1$ for all $i \in [n]$ is without loss of generality since one can appropriately rescale the values.} If $p_j$ is the price attached to item $j$, the \emph{revenue} is defined as $\sum_{j=1}^k p_{j}$. We denote it by $\revenue(\vm)$; it is fully specified by the vector of multipliers.

An autobidding instance $I$ comprises a set of autobidders $N = [n]$, a set of items $S = [k]$, and a valuation profile $\{ v_{i j} \}_{i \in N, j \in S}$. We will write $I = (N, S, \{ v_{i j} \}_{i \in N, j \in S})$ to denote such an instance.

We continue by recalling the notion of the \emph{price of anarchy (PoA)}. We let $\mathcal{I}$ be a set of autobidding instances and $E$ be a solution concept, so that $E(I)$ represents the set of solutions of instance $I \in \mathcal{I}$ per the solution concept $E$; we restrict to solution concepts that always exist, so that these statements are not vacuous. $E(I)$ returns a tuple $(\vm, \vx)$, where $\vm$ is the vector of multipliers and $\vx$ is the allocation vector. In this context, the price of anarchy of a solution concept $E$ in terms of the liquid welfare is defined as
\begin{equation}
    \label{eq:PoA}
    \poa_E \defeq \sup_{I \in \mathcal{I}, (\cdot, \vx) \in E(I)} \frac{ \sup_{\vec{x}^*} \welfare(\vec{x}^*)}{ \welfare(\vec{x}) }.
\end{equation}

The most natural solution concept in the presence of RoS constraints, and the one we focus on, is the (second-price) \emph{autobidding equilibrium}~\citep{Li24:Vulnerabilities}, introduced next.

\paragraph{Autobidding equilibrium} The first attempt at defining an autobidding equilibrium is as the pure Nash equilibria of the game where each bidder's $i \in [n]$ utility function is the value $\sum_{j=1}^k x_{i j} v_{i j}$ subject to the RoS constraint $\sum_{j=1}^k x_{i j} v_{i j} \geq \sum_{j=1}^k x_{i j} p_{j}$. However, \citet{Li24:Vulnerabilities} observed that, no matter how ties are broken, such a game can have no pure Nash equilibria. Instead, they put forward an alternative definition that allows for some additional flexibility in handling ties; we refer to~\citet{Mehta22:Auction} for a closely related notion.\footnote{A more precise nomenclature would be \emph{second-price} autobidding equilibria, since those equilibrium concepts can be defined beyond second-price auctions. For the sake of exposition, since we restrict our attention to second-price auctions, we use the term ``autobidding equilibria'' throughout the paper without specifying the underlying auction format. It is also worth pointing out that the term ``pacing equilibrium'' has also been used to refer to~\Cref{def:autobidequil}~\citep{Filos24:Ppad}.}

\begin{definition}[Autobidding equilibrium; \citealp{Li24:Vulnerabilities}]
    \label{def:autobidequil}
    An \emph{autobidding equilibrium} $(\vec{m}, \vec{x})$ (with respect to parallel second-price auctions) comprises a vector of multipliers $\vec{m} \in [1, M]^n$ and allocation vectors $\vec{x} \in [0, 1]^{n \times k}$ such that the following properties hold:
    \begin{enumerate}
        \item Only autobidders with the highest bid can win the item: if $x_{i j} > 0$, then $m_{i} v_{i j} = \max_{i' \in [n]} m_{i'} v_{i' j}$ for all $i \in [n]$ and $j \in [k]$.
        \item A winner pays the second highest bid: if $x_{i j} > 0$, then $p_j = \max_{i' \neq i} m_{i'} v_{i' j}$ for all $i \in [n]$ and $j \in [k]$.
        \item All items are fully allocated: $\sum_{i=1}^n x_{ij} = 1$ for all $j \in [k]$.
        \item Autobidders satisfy the RoS constraint: $\sum_{j=1}^k x_{i j} v_{i j} \geq \sum_{j=1}^k x_{i j} p_j$ for all $i \in [n]$.\label{item:RoS}
        \item Maximal pacing: unless $m_i = M$, $\sum_{j=1}^k x_{i j} v_{i j} = \sum_{j=1}^k x_{i j} p_j$ for all $i \in [n]$.\label{item:max-bid}
    \end{enumerate}
\end{definition}
Imposing an upper bound $M$ on the multipliers, as in the above definition, is necessary to guarantee existence; otherwise, an autobidder without a binding RoS constraint would be inclined to increase its multiplier indefinitely. 

\begin{remark}[Reserve prices]
    For convenience, some of our gadgets make use of \emph{reserve prices}, which can be thought of as having an additional autobidder bidding at that fixed price. As observed by~\citet{Li24:Vulnerabilities}, this can always be simulated via auxiliary autobidders (\Cref{appendix:reserveprices}). 
\end{remark}

A basic fact that we will use in our reduction is that autobidding equilibria are scale invariant:

\begin{claim}[Scale invariance] 
    \label{claim:scale-inv}
    If $I'$ is an instance obtained from $I$ by rescaling each value by $\eta > 0$, then $(\vm, \vx)$ is an autobidding equilibrium of $I$ if and only if it is an autobidding equilibrium of $I'$. 
\end{claim}

An autobidding equilibrium always exists, and computing one lies in the complexity class~$\FIXP$, which characterizes the complexity of Nash equilibria in multi-player general-sum games~\citep{Etessami10:Complexity}.

\begin{proposition}[\citealp{Filos24:Ppad,Li24:Vulnerabilities}]
    \label{prop:existence}
    An autobidding equilibrium always exists. Moreover, computing one is in $\FIXP$.
\end{proposition}

Interestingly, \citet{Filos24:Ppad} also observed that an \emph{exact} autobidding equilibrium per~\Cref{def:autobidequil} can be supported on irrational numbers.\footnote{To our knowledge, it is not known whether computing an approximate autobidding equilibrium per~\Cref{def:approx-autobidequil} is in \PPAD, although we suspect this is the case.} This motivates the following relaxation. (We have made some stylistic adjustments compared to the definition of~\citealp{Li24:Vulnerabilities}.)

\begin{definition}[\citealp{Li24:Vulnerabilities}]
    \label{def:approx-autobidequil}
    A \emph{$\beta$-approximate autobidding equilibrium} $(\vec{m}, \vec{x})$ (with respect to parallel second-price auctions) comprises a vector of multipliers $\vec{m} \in [1, M]^n$ and allocation vectors $\vec{x} \in [0, 1]^{n \times k}$ such that the following properties hold:
    \begin{enumerate}
        \item Only autobidders close enough to the highest bid can win the item: if $x_{i j} > 0$, then $m_{i} v_{i j} \geq (1 - \beta) \max_{i' \in [n]} m_{i'} v_{i' j}$ for all $i \in [n]$ and $j \in [k]$.
        \item A winner pays the second highest bid: if $x_{i j} > 0$, then $p_j = \max_{i' \neq i} m_{i'} v_{i' j}$ for all $i \in [n]$ and $j \in [k]$.
        \item All goods are fully allocated: $\sum_{i=1}^n x_{ij} = 1$ for all $j \in [k]$.
        \item Autobidders approximately satisfy the RoS constraint: $\sum_{j=1}^k x_{i j} v_{i j} \geq (1 - \beta) \sum_{j=1}^k x_{i j} p_j$ for all $i \in [n]$.
        \item Approximately maximal pacing: unless $m_i = M$, $\sum_{j=1}^k x_{i j} v_{i j} \leq (1+\beta) \sum_{j=1}^k x_{i j} p_j$ for all $i \in [n]$.
    \end{enumerate}
\end{definition}

A basic fact, which goes back to the early paper of~\citet{Aggarwal19:Autobidding}, is that the price of anarchy with respect to autobidding equilibria is always at most $2$; for completeness, we include the proof in~\Cref{appendix:proofs2} since~\citet{Aggarwal19:Autobidding} did not explicitly used the language of autobidding equilibria in their treatment. (The result of~\citet{Aggarwal19:Autobidding} is in fact significantly more general.)

\begin{restatable}[\citealp{Aggarwal19:Autobidding}]{proposition}{poatwo}
    \label{prop:poa-ub}
    The price of anarchy~\eqref{eq:PoA} with respect to autobidding equilibria is at most $2$.
\end{restatable}

\paragraph{Computing high-quality autobidding equilibria} We are interested in the following computational problem. Given an autobidding instance $I$, let $\opt$ be the optimal liquid welfare attained by an autobidding equilibrium. The goal is to compute an autobidding equilibrium $(\vec{m}, \vx)$ such that
\begin{equation}
    \label{eq:approx-wel}
    \welfare(\vx) \geq \frac{1}{\rho} \cdot \opt
\end{equation}
for some parameter $\rho \geq 1$.\footnote{In the literature on approximation algorithms, it is more common to use $\rho^{-1}$ instead of $\rho$ when it comes to maximization problems, but we use~\eqref{eq:approx-wel} in order to relate directly the approximation factor with the PoA bound in~\eqref{eq:PoA}.} \Cref{prop:poa-ub} implies that for $\rho = 2$, this reduces to finding any autobidding equilibrium; \eqref{eq:approx-wel} would then automatically be met since $\opt \leq \sup_{\vx^*} \welfare(\vx^*)$. We call attention to the fact that~\eqref{eq:approx-wel} measures approximation in terms of the optimal autobidding equilibrium, whereas the price of anarchy bound~\eqref{eq:PoA} compares relative to the optimal allocation unfettered by equilibrium constraints. Similarly, the approximation objective~\eqref{eq:approx-wel} can be adapted in terms of the revenue. Prior work by~\citet{Li24:Vulnerabilities} has shown that both problems are \APX-hard, meaning that they are hard to approximate for \emph{some} constant $\rho$. In particular, this means that there is no polynomial-time approximation scheme (\PTAS).

\begin{remark}
    Our reduction in~\Cref{sec:inapprox-stable} readily carries over under approximate autobidding equilibria per~\Cref{def:approx-autobidequil}; we do not spell this out as it unnecessarily complicates the arguments. Our analysis in~\Cref{sec:learning} explicitly accounts for approximation errors.
\end{remark}

We also examine a decision version of this problem, where the question is whether there exists an autobidding equilibrium whose welfare is within a certain factor $\rho \geq 1$ from the optimal $\sup_{\vx^*} \welfare(\vx^*)$. In view of~\Cref{prop:existence}, this question is trivial when $\rho \geq 2$, so the interesting regime is when $\rho < 2$.
\section{Inapproximability of high-quality autobidding equilibria}
\label{sec:inapprox-stable}

In this section, we establish inapproximability results for computing autobidding equilibria with high (liquid) welfare (\Cref{theorem:welfare-hard}) or revenue (\Cref{theorem:revenue-hard}). We then go on to refine our hardness results in the presence of ML advice (\Cref{sec:MLadvise}).

\paragraph{Label cover} The starting point of our reduction is the \emph{label cover problem}~\citep{Arora97:Hardness}, which is a mainstay in the hardness of approximation literature.

\begin{definition}[Label cover]
    \label{def:labelcover}
    A \emph{label cover} instance consists of a tuple $(G, \Sigma, \Pi)$, where
    \begin{itemize}
        \item $G = (V_1, V_2, E)$ is a left- and right-regular bipartite graph with vertex sets $V_1$ and $V_2$ and edge set $E$;
        \item every vertex in $V_1 \cup V_2$ is to be assigned a label from the finite alphabet $\Sigma$;\footnote{It is common to assume that vertices in $V_1$ are labeled from an alphabet $\Sigma_1$ and vertices in $V_2$ from a disjoint alphabet $\Sigma_2$, but this distinction is not important for our purposes.} and
        \item for each edge $e \in E$, there is a constraint $\Pi_e: \Sigma \to \Sigma$ (the \emph{projection property}) so that
        \begin{equation*}
            \Pi \defeq \left\{ \Pi_e : \Sigma \to \Sigma, e \in E  \right\}.
        \end{equation*}
    \end{itemize}
    A \emph{labeling} of $G$ is a mapping $\ell : V \to \Sigma$, assigning a label to each vertex. We say that $\ell$ \emph{satisfies} an edge $e = (u, v) \in E$ if $\Pi_e(\ell(u)) = \ell(v)$.
\end{definition}

A well-studied, more constrained version of this problem---known as the \emph{unique} label cover problem---postulates that each $\Pi_e$ is a permutation~\citep{Khot02:Power}, but our reduction does not require that particular structure. We now introduce the \emph{promise}, gap-amplified version of label cover.

\begin{definition}
    For a parameter $\delta \in (0, 1)$, the problem $\gaplabel_{1, \delta}$ is defined with respect to an instance of label cover $(G, \Sigma, \Pi)$ such that either of the following holds:
    \begin{itemize}
        \item There exists a labeling $\ell : V \to \Sigma$ that satisfies \emph{all} the edge constraints.
        \item Any labeling $\ell : V \to \Sigma$ satisfies \emph{at most} $\delta |E|$ of the edge constraints.
    \end{itemize}
    $\gaplabel_{1, \delta}$ is the decision problem of identifying which of the above two cases is true, under the promise that one of them holds.
\end{definition}

We will use the following seminal hardness result concerning $\gaplabel_{1, \delta}(\Sigma)$~\citep{Raz98:Parallel,Arora98:Proof}.

\begin{theorem}    \label{theorem:labelcover}
    For any constant $\delta > 0$, there is an alphabet $\Sigma = \Sigma(\delta)$ such that $\gaplabel_{1, \delta}$ is \NP-hard.
\end{theorem}

For our purposes, the regime where $\delta = o(1)$ is also relevant. We will use the following result due to~\citet{Moshkovitz08:Two}.

\begin{theorem}[\citealp{Moshkovitz08:Two}]
    \label{remark:alphabetsize}
    There exists a constant $c > 0$ such that for every $\delta \geq 1/N^c$, \textsc{SAT} on input of size $n$ can be efficiently reduced to a label cover instance of size $N = n^{1 + o(1)} \poly(1/\delta)$ over an alphabet of size $\exp(1/\delta)$ and soundness error $\delta$.
\end{theorem}

There is a common conjecture in this line of work that asserts that it suffices to take the alphabet size as $|\Sigma| \leq \poly(1/\delta)$~\citep{Moshkovitz15:Projection} (\emph{cf.}~the sliding scale conjecture of~\citealp{Bellare93:Efficient}).

\begin{conjecture}[Projection games conjecture; \citealp{Moshkovitz15:Projection}]
    \label{conj:projection}
    There exists a constant $c > 0$ such that for every $\delta \geq 1/N^c$, \textsc{SAT} on input of size $n$ can be efficiently reduced to a label cover instance of size $N = n^{1 + o(1)} \poly(1/\delta)$ over an alphabet of size $\poly(1/\delta)$ and soundness error $\delta$.
\end{conjecture}

We will make use of this conjecture to obtain (conditionally) stronger polynomial inapproximability for revenue (\Cref{cor:revenue-projectiongames}).

Before we proceed, it will be useful to reformulate the label cover problem as follows. For every vertex $u \in V \defeq V_1 \cup V_2$, we introduce a Boolean variable $z_{u, \sigma} \in \{0, 1\}$ for each label $\sigma \in \Sigma$ subject to the constraint $\sum_{\sigma \in \Sigma} z_{u,\sigma} = 1$; this enforces that exactly one of the variables $\{ z_{u, \sigma} \}_{\sigma \in \Sigma}$ will be active, ensuring a legitimate labeling. The optimization problem underpinning the label cover problem can then be phrased as the constraint satisfaction problem (CSP)
\begin{equation}
    \label{eq:labelcover-CSP}
    \sum_{(u,v) = e \in E} \left( \bigvee_{\sigma \in \Sigma} ( z_{u,\sigma} \land z_{v, \Pi_e(\sigma)}) \right)
\end{equation}
subject to the uniqueness constraint $\sum_{\sigma \in \Sigma} z_{u,\sigma} = 1$ for all $u \in V$.\footnote{Throughout this paper, we associate \textsf{True} with $1$ and \textsf{False} with $0$, so that~\eqref{eq:labelcover-CSP} amounts to the number of satisfied edge constraints.} In particular, if $\Psi_e(\vec{z}) \defeq \bigvee_{\sigma \in \Sigma} ( z_{u, \sigma} \land z_{v,\Pi_e(\sigma)} )$ for $(u, v) = e$ is the edge constraint, one has to assign the Boolean variables $\vec{z}$ to satisfy as many of these edge constraints as possible. An obvious fact used in our reduction is that allowing $\sum_{\sigma \in \Sigma} z_{u, \sigma} \leq 1$ can never increase the number of satisfied edge constraints. In what follows, it will be convenient to associate the alphabet $\Sigma$ with the set $\{1, \dots, |\Sigma| \}$.

We split our reduction into a series of basic gadgets. \Cref{sec:labelassign} shows how to associate each variable to a multiplier in a way that corresponds to valid label assignment. \Cref{sec:AND} implements the conjunction required to encode~\eqref{eq:labelcover-CSP}. We then use these gadgets to construct an autobidding instance in which the welfare or revenue are tied to the number of satisfied edge constraints. (The high-level overview of our reduction was illustrated earlier in~\Cref{fig:diagram}.)

\subsection{Label assignment}
\label{sec:labelassign}

The first step in the reduction is to assign a label to each variable. To do so, for each vertex $u \in V$, we consider $|\Sigma|$ autobidders competing for $|\Sigma|+1$ items. The valuations are defined as follows.
\begin{equation}
    \label{eq:labelassign}
    \frac{1}{\eta} v_{i j} = 
    \begin{cases}
        M & \text{if } i = j,\\
        1 & \text{if } i \neq j \text{ and } j \neq |\Sigma| + 1,\\
        K & \text{if } j = |\Sigma|+1.
    \end{cases}    
\end{equation}
Above, $\eta \ll 1$ and $K \gg M$ are parameters that will be chosen appropriately. Specifically, $\eta$ is a scaling parameter that does not affect autobidding equilibria (\Cref{claim:scale-inv}). We will eventually select $\eta$ to be small enough so that the welfare and revenue of the gadgets are negligible. This valuation profile is illustrated in~\Cref{tab:assign}.
\begin{table}[t]
    \centering
    \caption{The construction behind~\Cref{lemma:assignment}, up to a rescaling factor $\eta$ that does not affect the autobidding equilibria. The key property that guarantees soundness is $K \gg M$.}
    \begin{tabular}{c c c c c c}
        & Item 1 & Item 2 & $\cdots$ & Item $|\Sigma|$ & Item $|\Sigma|+1$ \\ 
        Autobidder 1 & $M$ & $1$ & $\cdots$ & $1$ & $K $ \\
        Autobidder 2 & $1$ & $M$ & $\cdots$ & $1$ & $K$ \\
        $\vdots$ & $\vdots$ & $\vdots$ & $\ddots$ & $\vdots$ & $\vdots$ \\
        Autobidder $|\Sigma|$ & $1$ & $1$ & $\cdots$ & $M$ & $K$ \\
    \end{tabular}
    \label{tab:assign}
\end{table}

In this instance, each autobidder $i$ values highly the corresponding item $i$, but by far the most valuable item is the last one, which is equally valued by all autobidders. The role of this highly valued item is that it precludes having more than one autobidder select a large multiplier, for otherwise the RoS constraint would be violated. This is formalized below; the proofs from this section are deferred to~\Cref{appendix:proofs-main}.

\begin{restatable}[Label assignment]{lemma}{labass}
    \label{lemma:assignment}
    Consider the valuation profile given in~\eqref{eq:labelassign} with respect to some vertex $u \in V$ and $K$ large enough. For any label $\sigma \in \Sigma$, the vector of multipliers $(m_{u, \sigma'})_{\sigma' \in \Sigma} $ in which $m_{u,\sigma} = M$ and $m_{u, \sigma'} = 1$ for all $\sigma' \neq \sigma$ is part of an autobidding equilibrium. Furthermore, any autobidding equilibrium has to satisfy either of the following two properties:
    \begin{itemize}
        \item the maximum multiplier is at most $(M |\Sigma| + K)/(|\Sigma| + K) \approx 1$ when $K \gg M$, or
        \item the maximum multiplier is $M$ and the second highest multiplier is at most $1 + (M + |\Sigma|)/K \approx 1$ when $K \gg M$.
    \end{itemize}
\end{restatable}

In view of~\Cref{lemma:assignment}, we associate a multiplier that has value $M$ with $\textsf{True}$ and a multiplier that has value $\approx 1$ with \textsf{False}. Ultimately, we will set $M \gg 1$. \Cref{lemma:assignment} tells us that at most one of the variables can be activated, which simulates the process of label assignment---\Cref{lemma:assignment} leaves open the possibility that no label is given.

The next subsections show how to implement an \textsf{AND} gadget, which is needed in order to encode the conjunctions in~\eqref{eq:labelcover-CSP}. This is accomplished in two steps: we first implement a \textsf{NOT} gadget, and we then complete the construction through a \textsf{NAND} gadget. 

When we glue together multiple autobidding systems, we will maintain a basic invariance that guarantees that we can analyze those systems independently; we formalize this below.

\begin{restatable}[Conservative extension of autobidding instances]{lemma}{consext}
    \label{lemma:extension}
    Consider two instances $I = (N, S, \cdot)$ and $I' = (N', S', \cdot)$ with $S \subset S'$ and $N \subset N'$. If in instance $I'$, every autobidder $i \in N$ never obtains a positive fraction of an item in $S' \setminus S$ and every autobidder $i \in N' \setminus N$ never obtains a positive fraction of an item in $S$, no matter the choice of multipliers in $[1, M]^{N + N'}$, then an autobidding equilibrium $(\vec{m}, \vec{x})$ of $I'$ is such that $((m_i)_{i \in N}, (x_{i j})_{i \in N, j \in S} )$ is an autobidding equilibrium of $I$.
\end{restatable}

In other words, if we have established a certain property for autobidding equilibria in the original instance $I$, this property is maintained in the extended instance $I'$. An instance $I'$ satisfying the preconditions of~\Cref{lemma:extension} will be referred to as \emph{conservative extension} of $I$.

\subsection{AND gadget}
\label{sec:AND}

The next step in our reduction is implementing an \textsf{AND} gadget, which is done in two steps.

\paragraph{NAND gadget} We begin by constructing an \textsf{NAND} gadget in the following formal sense.

\begin{table}[t]
    \centering
    \caption{The construction behind~\Cref{lemma:nand}, up to a rescaling factor $\eta$ that does not affect the autobidding equilibria.}
    \begin{tabular}{c c c c c}
    & Item 1 & Item 2 & Item 3 & Item 4 \\ 
     Input autobidder 1 & $1/(2M)$ & $0$ & $0$ & $0$ \\
     Input autobidder 2 & $0$ & $1/(2M)$ & $0$ & $0$ \\
     Output autobidder & $1/2 + \epsilon$ & $1/2 + \epsilon$ & $1$ & $1/(2M)$\\
     \bottomrule
     Reserve price & --- & --- & $1 + 3 \epsilon$ & $0.5$
\end{tabular}
    \label{tab:nand}
\end{table}

\begin{restatable}[Negation of conjunction]{lemma}{nandgate}
    \label{lemma:nand}
    Let $m \in [1, 1+\epsilon] \cup \{M\}$ and $m' \in [1, 1 + \epsilon] \cup \{M\}$ be the multipliers of two autobidders for some $\epsilon > 0$. There is a conservative extension that produces an output autobidder whose multiplier $\widebar{m}$ satisfies, in equilibrium,
    \begin{equation*}
        \widebar{m} \in 
        \begin{cases}
            [1, 1 + 3 \epsilon] & \text{if } m = M \text{ and } m' = M, \\
            \{M\} & \text{otherwise}.
        \end{cases}
    \end{equation*}
\end{restatable}

In the above construction (illustrated in~\Cref{tab:nand}), we used reserve prices. As observed by~\citet{Li24:Vulnerabilities}, there is a basic gadget with two auxiliary autobidders and two auxiliary items that simulates the presence of reserve prices (\Cref{appendix:reserveprices}).

\paragraph{NOT gadget} To complete the construction of the \textsf{AND} gadget, we now show how to implement a negation gadget (\Cref{tab:neg}) in the following formal sense.

\begin{restatable}[Negation]{lemma}{negation}
    \label{lemma:negation}
    Let $m \in [1, 1 + \epsilon] \cup \{M\}$ be the multiplier of an input autobidder. There is a conservative extension that produces an output autobidder whose multiplier satisfies, in equilibrium,
    \begin{equation*}
        \widebar{m} \in 
        \begin{cases}
            [1, 1 + 2 \epsilon] & \text{if } $m = M$,\\
            \{M\} & \text{otherwise}.
        \end{cases}
    \end{equation*}
\end{restatable}

\begin{table}[t]
    \centering
    \begin{tabular}{c c c c}
    & Item 1 & Item 2 & Item 3 \\
     Input autobidder & $1/M$ & $0$ & $0$ \\
     Output autobidder & $1 + \epsilon$ & $1$ & $1/M$\\
     \bottomrule
     Reserve price & --- & $1 + 2 \epsilon$ & $1$
\end{tabular}
    \caption{The construction behind~\Cref{lemma:negation}, up to a rescaling factor $\eta$ that does not affect the autobidding equilibria.}
    \label{tab:neg}
\end{table}

\subsection{Putting everything together: encoding the label cover CSP}

Combining the gadgets of~\Cref{lemma:assignment,lemma:nand,lemma:negation}, we construct for each edge $e \in E$ and label $\sigma \in \Sigma$ an autobidder whose multiplier $m_{e, \sigma}$ satisfies, in equilibrium,
\begin{equation}
    \label{eq:edge-sigma}
 m_{e, \sigma} \in
 \begin{cases}
     \{M\} & \text{if } m_{u,\sigma} = M \text{ and } m_{v, \Pi_e(\sigma)} = M, \\
     [1, 1 + \epsilon] & \text{otherwise}.
 \end{cases}
\end{equation}
This is done by first employing the construction of~\Cref{lemma:nand} to obtain $\widebar{m}_{e, \sigma}$ and then~\Cref{lemma:negation} to obtain $m_{e, \sigma}$, for each $e \in E$ and $\sigma \in \Sigma$. When $K = K(\epsilon)$ for some parameter $\epsilon > 0$, \Cref{lemma:assignment,lemma:extension,lemma:nand,lemma:negation} imply the validity of~\eqref{eq:edge-sigma}. More precisely, for any $K \geq M |\Sigma|/\epsilon'$, \Cref{lemma:assignment} implies that every multiplier in the label assignment satisfies $m_{u, \sigma} \in [1, 1 + \epsilon'] \cup \{M\}$. Accounting for the error in~\Cref{lemma:nand,lemma:negation}, it follows that $m_{e, \sigma} \in [1, 1 + 6 \epsilon'] \cup \{M \}$. This implies that setting $K = 6 M |\Sigma|/\epsilon$ guarantees~\eqref{eq:edge-sigma}.

\paragraph{The edge autobidders} The next step is to create a new autobidder for each edge $e \in E$, which we refer to as the \emph{edge autobidder}. Each such edge autobidder competes with the autobidders associated with the multipliers $\{ m_{e, \sigma} \}_{\sigma \in \Sigma}$ for a single item, the \emph{edge item}. The edge autobidder values that item $1 + \epsilon$ whereas all other autobidders value that item $1/M$. This guarantees that it is always the edge autobidder that secures that item---maintaining the basic invariance required by~\Cref{lemma:extension}. The price paid by the edge autobidder is $1$ if $m_{e, \sigma} = M$ for \emph{at least some} $\sigma \in \Sigma$, and at most $(1 + \epsilon)/M$ otherwise. In other words, this construction ensures that the satisfaction of the edge constraint associated with $e \in E$ is encoded in the price paid by the corresponding autobidder (\Cref{tab:edgeautobidder}).

\begin{table}[t]
    \centering
    \caption{The interaction of the edge autobidder. The edge autobidder always secures the item, but the price is high if \emph{at least one} of the competing autobidders selects a high multiplier.}
    \begin{tabular}{c c}
        & Edge item \\ 
        Edge autobidder & $1+\epsilon$ \\
        Autobidder 1 & $1/M$ \\
        Autobidder 2 & $1/M$ \\ $\vdots$ & $\vdots$ \\
        Autobidder $|\Sigma|$ & $1/M$ \\
    \end{tabular}
    \label{tab:edgeautobidder}
\end{table}

In what follows, we recall that $\eta > 0$ is the scaling factor used in the constructions of~\Cref{lemma:nand,lemma:negation,lemma:assignment}. This ensures that the revenue and welfare derived from those gadgets are negligible. Formally, let $(\vec{m}, \vec{x})$ be an autobidding equilibrium of the resulting instance, induced by a given label cover instance. We first have
\begin{equation}
    \label{eq:revenue-ub}
    \revenue(\vm) \leq \eta \left( |V| ( |\Sigma| M + K ) + 5 |E| |\Sigma| \right) + \sum_{e \in E}  \bigvee_{\sigma \in \Sigma} ( m_{e,\sigma} = M ) + \frac{|E| ( 1 + \epsilon) }{M}.
\end{equation}
In particular, the term $\eta |V| (|\Sigma| M + K)$ is a bound on the total revenue from the $|V|$ label assignment gadgets; the term $\frac{5}{2} \eta |E| |\Sigma|$ is a bound on the total revenue from the $|E| |\Sigma|$ \textsf{NAND} gadgets; and the term $\frac{5}{2} \eta |E| |\Sigma|$ is a bound on the total revenue from the $|E| |\Sigma|$ \textsf{NOT} gadgets. Furthermore, for every edge $e \in E$, the edge autobidder pays $1$ for the edge item if $ \bigvee_{\sigma \in \Sigma} ( m_{e,\sigma} = M ) = 1$ and at most $(1 + \epsilon)/M$ otherwise. This justifies the revenue upper bound given in~\eqref{eq:revenue-ub}. 

Moreover, by completeness of the label assignment gadget (\Cref{lemma:assignment}), it follows that there exists an autobidding equilibrium $(\vm, \vx)$ such that
\begin{equation}
    \label{eq:revenue-lb}
    \revenue(\vm) \geq \sum_{e \in E}  \bigvee_{\sigma \in \Sigma} ( m_{e, \sigma} = M ).
\end{equation}
Setting the parameters appropriately (\Cref{appendix:proofs-main}) and combining with~\Cref{theorem:labelcover}, we arrive at the conclusion that approximating the optimal revenue to any constant factor is \NP-hard.

\begin{restatable}{theorem}{revenuehard}
    \label{theorem:revenue-hard}
    Computing an autobidding equilibrium that approximates the revenue-optimal one within any constant factor $\rho \geq 1$ is \NP-hard.
\end{restatable}

The instance we have constructed contains $O(|E||\Sigma| )$ autobidders and items. In light of the fact that one can take the alphabet size to be $|\Sigma| \leq \exp(1/\delta)$ (\Cref{remark:alphabetsize}), our reduction also implies logarithmic inapproximability with respect to revenue.

\begin{corollary}
    Computing an autobidding equilibrium that approximates the revenue-optimal one within some logarithmic factor $\log^{\Omega(1)}(n k)$ is \NP-hard.
\end{corollary}

Moreover, assuming~\Cref{conj:projection}, we can take $|\Sigma| = \poly(1/\delta)$. If $\delta = 1/|E|^c$ for some constant $c > 0$, we obtain an $\Omega(\poly(|E|))$ gap in the revenue and the instance has size $\poly(|E|)$.

\begin{corollary}
    \label{cor:revenue-projectiongames}
    Under the projection games conjecture (\Cref{conj:projection}), computing an autobidding equilibrium that approximates the optimal one within some polynomial factor $(nk)^{\Omega(1)}$ is \NP-hard.
\end{corollary}

\paragraph{Welfare} We continue by reducing the $\gaplabel$ problem to finding autobidding equilibria with high welfare. In addition to the instance we have constructed to obtain~\Cref{theorem:revenue-hard}, we consider, for each edge $e \in E$, an additional autobidder, which will be referred to as the \emph{incumbent edge autobidder}. Each incumbent edge autobidder competes with the corresponding edge autobidder for a single item, the \emph{incumbent edge item}. The edge autobidder values the incumbent edge item $(1 + \epsilon)/M$ while the incumbent edge autobidder values it $1$ (\Cref{tab:inc-edge}). Efficiency would dictate allocating this item to the incumbent edge autobidder. However, when the edge autobidder has a surplus (from the interaction described in \Cref{tab:edgeautobidder}), it can usurp that item from the incumbent edge autobidder. This is formalized below.

\begin{table}[t]
    \centering
    \caption{The interaction of the edge autobidder with the incumbent edge autobidder.}
    \begin{tabular}{c c}
        & Incumbent edge item \\ 
        Edge autobidder & $(1+\epsilon)/M$ \\
        Incumbent edge autobidder & $1$ \\
    \end{tabular}
    \label{tab:inc-edge}
\end{table}

\begin{restatable}{lemma}{stealing}
    \label{lemma:stealing}
    Fix an edge $e \in E$. If $m_{e, \sigma} = M$ for some $\sigma \in \Sigma$, the incumbent edge autobidder wins at least a $(1 - 2 \epsilon)$ fraction of the incumbent edge item in equilibrium. Otherwise, if $m_{e, \sigma} \leq 1 + \epsilon$ for all $\sigma \in \Sigma$, it is the edge autobidder who exclusively wins the incumbent edge item.
\end{restatable}

We are now ready to analyze the (liquid) welfare in an autobidding equilibrium $(\vm, \vx)$ of the instance that has arisen from our reduction. We first argue that
\begin{equation}
    \label{eq:wel-ub}
    \welfare(\vx) \leq \eta \left( |V| ( |\Sigma| M + K ) + 5 |E| |\Sigma| \right) + (1+\epsilon) |E| + \sum_{e \in E} \bigvee_{\sigma \in \Sigma} ( m_{e, \sigma} = M ) + \frac{|E| ( 1 + \epsilon) }{M}.
\end{equation}
The first term in the right-hand side follows as in~\eqref{eq:revenue-lb}; the second term, $(1+\epsilon) |E|$, is the total value obtained by the edge autobidders for winning the edge items; the third term corresponds to the value obtained by the incumbent edge autobidders: by~\Cref{lemma:stealing}, if $m_{e, \sigma} = M$ for some $\sigma \in \Sigma$, the corresponding incumbent edge autobidder wins a $(1 - 2 \epsilon)$ fraction of the incumbent edge item; for every other edge, \Cref{lemma:stealing} implies that the edge autobidder wins that item, for a much lower value of $(1 + \epsilon)/M$, giving rise to the last term in~\eqref{eq:wel-ub}. Similarly, by~\Cref{lemma:stealing}, there exists an autobidding equilibrium $(\vm, \vx)$ such that
\begin{equation*}
    \welfare(\vx) \geq (1 +\epsilon) |E| + ( 1 - 2 \epsilon) \sum_{e \in E}  \bigvee_{\sigma \in \Sigma} ( m_{e, \sigma} = M ).
\end{equation*}
As a result, setting the parameters appropriately yields the following.

\begin{theorem}
    \label{theorem:welfare-hard}
    For any constant $\epsilon > 0$, it is \NP-hard to compute an autobidding equilibrium that approximates the welfare-optimal one within a factor $2 - \epsilon$.
\end{theorem}

Turning now to a decision version of the problem, it follows that ascertaining whether a non-trivial autobidding equilibrium exists is \NP-hard.

\begin{corollary}
    For any constant $\epsilon > 0$, it is \NP-hard to decide whether there exists an autobidding equilibrium whose welfare is at least $1/2 + \epsilon$ that of the optimal allocation.
\end{corollary}

\subsection{Refined inapproximability under ML advice}
\label{sec:MLadvise}

One way of improving PoA bounds is through ML advice concerning the underlying valuations~\citep{Balseiro21:Robust,Deng24:Efficiency}. In this subsection, we examine how our inapproximability results can be adapted in the presence of such side information. Specifically, following~\citet{Balseiro21:Robust}, we work under the following model.

\begin{definition}[Approximate signals; \citealp{Balseiro21:Robust}]
    \label{def:approx-signal}
    The auctioneer has access to a signal $\signal_{i j}$, for all autobidders $i \in [n]$ and items $j \in [k]$, such that $\signal_{i j} \in [\gamma \cdot v_{i j}, v_{i j}]$ for some parameter $\gamma \in [0, 1)$; this multiplicative guarantee is referred to as a \emph{$\gamma$-approximate signal}.
\end{definition}
\citet{Balseiro21:Robust} proposed to use $\signal_{i j}$ as a reserve price, and were able to obtain PoA bounds parameterized based on $\gamma$, the accuracy of the signal.

We show that our inapproximability results can be parameterized based on the value of $\gamma$. We first claim that~\Cref{lemma:assignment,lemma:nand,lemma:negation} readily extend when one uses $\signal_{i j}$ as a reserve price (\Cref{appendix:MLadvise}).\footnote{If an item already had a reserve price attached to it, its new reserve price is the maximum of the two.} Concentrating first on the revenue, what changes now is that the optimal revenue of an autobidding equilibrium satisfies
\begin{equation*}
    \revenue \approx (1 - \gamma) \sum_{e \in E}  \bigvee_{\sigma \in \Sigma} ( m_{e, \sigma} = M ) + \gamma |E|
\end{equation*}
when the parameters of the reduction are set appropriately (\Cref{appendix:proofs-main} contains the precise argument). The reason for this is that now each edge autobidder has to pay at least $(1+\epsilon) \gamma$, which is the reserve price, to obtain the edge item even when the associated edge clause is unsatisfiable. As before, the edge autobidder pays $1$ when the edge clause is satisfiable. We thus arrive at the following refinement of~\Cref{theorem:revenue-hard}.

\begin{restatable}{theorem}{gammarevenue}
    Under reserve prices set to $\gamma$-approximate signals, it is \NP-hard to compute an autobidding equilibrium that approximates the optimal one within a factor $1/(\gamma + \epsilon)$, for any constant $\epsilon > 0$.
\end{restatable}

This again approaches the PoA bound (with respect to revenue) of $1/\gamma$ established by~\citet{Balseiro21:Robust}, revealing again a tight connection between PoA and the inapproximability threshold.

Next, we turn our attention to our welfare reduction. It suffices to adapt~\Cref{lemma:stealing} as follows.

\begin{restatable}{lemma}{refinedstealing}
    \label{lemma:refined-stealing}
    Fix an edge $e \in E$ and reserve prices set to $\gamma$-approximate signals. If $m_e^{(\sigma)} = M$ for some $\sigma \in \Sigma$, the incumbent edge autobidder wins at least a $(1 - 2\epsilon)$ fraction of the incumbent edge item. Otherwise, the edge autobidder wins at least a $1 - \gamma$ fraction of that item.
\end{restatable}

Using this lemma, we find that the optimal liquid welfare at an autobidding equilibrium satisfies
\begin{equation*}
    \welfare \approx (1 + \gamma) |E| +  (1 - \gamma) \sum_{e \in E}  \bigvee_{\sigma \in \Sigma} ( m_{e, \sigma} = M )
\end{equation*}
when the parameters of the reduction are set appropriately.
We are thus able to refine~\Cref{theorem:welfare-hard} as follows.
\begin{restatable}{theorem}{gammawelfare}
    Under reserve prices set to $\gamma$-approximate signals, it is \NP-hard to compute an autobidding equilibrium that approximates the optimal one within a factor $2/(1 + \gamma + \epsilon)$, for any constant $\epsilon > 0$.
\end{restatable}

This approaches the PoA bound (with respect to liquid welfare) of~\citet{Balseiro21:Robust} under reserve prices \emph{and} additive boosts. With only reserve prices set to approximate signals per~\Cref{def:approx-signal}, \citet{Balseiro21:Robust} established a weaker $2 - \gamma$ PoA bound. This leaves an interesting gap for future research.
\section{Inapproximability for learning dynamics}
\label{sec:learning}

Our foregoing results in~\Cref{sec:inapprox-stable} concern autobidding equilibria, which can be thought of as fixed points of learning dynamics. However, computing autobidding equilibria is \PPAD-hard in general~\citep{Li24:Vulnerabilities}, so convergence is not guaranteed~\citep{Leme24:Complex,Lucier24:Autobidders,Gaitonde23:Budget}. This section establishes inapproximability results under relaxed solution concepts that will be met by a broad class of learning algorithms.

\subsection{Revenue inapproximability under time-average RoS constraints}

We begin by showing that the mere presence of (average) RoS constraints suffices to drive \NP-hardness results. Specifically, we work under the following definition.

\begin{definition}
    \label{def:admissible}
    Consider a sequence of multipliers $\vec{m}^{(1)}, \dots, \vec{m}^{(T)} \in [1, M]^n$ and allocation vectors $\vec{x}^{(1)}, \dots, \vec{x}^{(T)} \in [0, 1]^{n 
    \times k}$. We say that $(\vec{m}^{(1)}, \vec{x}^{(1)}), \dots, (\vec{m}^{(T)}, \vec{x}^{(T)})$ is a \emph{second-price sequence} satisfying the \emph{$\beta$-approximate time-average RoS constraints} if the following properties hold:
    \begin{enumerate}
        \item In each round $t \in [T]$, only autobidders with the highest bid can win the item: if $x^{(t)}_{i j} > 0$, then $m^{(t)}_{i} v_{i j} = \max_{i' \in [n]} m^{(t)}_{i'} v_{i' j}$ for all $i \in [n]$ and $j \in [k]$.
        \item In each round $t \in [T]$, a winner pays the second highest bid: if $x^{(t)}_{i j} > 0$, then $p^{(t)}_j = \max_{i' \neq i} m^{(t)}_{i'} v_{i' j}$ for all $i \in [n]$ and $j \in [k]$.
        \item In each round $t \in [T]$, all items are fully allocated: $\sum_{i=1}^n x^{(t)}_{ij} = 1$ for all $j \in [k]$.
        \item Bidders approximately satisfy the RoS constraint on average across the $T$ rounds:\footnote{Compared to~\Cref{def:approx-autobidequil}, here we consider an additive approximation, since this is more in line with guarantees obtained under learning dynamics.}\label{item:average-RoS}
        \[
         \frac{1}{T} \sum_{t=1}^T \sum_{j=1}^k x^{(t)}_{i j} v_{i j} \geq \frac{1}{T} \sum_{t=1}^T \sum_{j=1}^k x^{(t)}_{i j} p^{(t)}_j - \beta \text{ for all } i \in [n].
        \]
    \end{enumerate}
\end{definition}

The welfare of a sequence $(\vm^{(1)}, \vx^{(1)}), \dots, (\vm^{(T)}, \vx^{(T)})$ is defined as the average welfare across the $T$ rounds, $\frac{1}{T} \sum_{t=1}^T \welfare(\vx^{(t)})$, and similarly for the revenue.

\paragraph{Maximum cover} The class of CSPs we consider in this section is based on the maximum cover problem (\textsc{MaxCover}). Here, we are given a universe $U$ and a collection $\mathcal{F} = \{S_1, \dots, S_p\} $ that contains subsets of $U$. The problem is to select $q \in \N$ of these subsets so as to cover as many elements of $U$ as possible. We will use the following classic result due to \citet{Feige98:Threshold}.

\begin{theorem}[\citealp{Feige98:Threshold}]
    \label{theorem:Feige}
    Given an instance of \textsc{MaxCover} with parameter $q \in \N$, it is \NP-hard to distinguish between the following two cases under the promise that one of them holds:
    \begin{itemize}
        \item there exists a $q$-cover that contains all elements of $U$, or
        \item any $q$-cover contains at most a $1 - 1/e + \delta$ fraction of the elements of $U$, for any constant $\delta > 0$.
    \end{itemize}
\end{theorem}

We now define a class of CSPs that captures \textsc{MaxCover}. We are given a set of variables $z_1, \dots, z_d$, each of which is to be assigned a label from a set $|\Sigma|$. We write $z_{i, \sigma} \in \{0, 1\}$ to denote whether $z_i$ was assigned the label $\sigma \in \Sigma$, so that $\sum_{\sigma \in \Sigma} z_{i, \sigma} = 1$. Every clause $C$ of the CSP is of the form $\bigvee_{(i, \sigma) \in C} z_{i, \sigma}$. We let $\mathcal{C}$ denote the set of clauses. The goal is to maximize $\sum_{C \in \mathcal{C}} \bigvee_{(i, \sigma) \in C} z_{i, \sigma}$ subject to $\sum_{\sigma \in \Sigma} z_{i, \sigma} = 1$.

This class of CSPs clearly encompasses \textsc{MaxSAT} when the alphabet is binary. It also encompasses \textsc{MaxCover}. This can be seen by using each of $q$ variables $z_i$ to denote through its label the corresponding subset from $\mathcal{F}$ that was selected (so $|\Sigma| = |\mathcal{F}|$), and each clause corresponding to an element in the universe; this encoding makes allowance for selecting the same subset more than once, but that clearly preserves soundness. This means that this class of CSPs is subject to the hardness result in~\Cref{theorem:Feige}.

\paragraph{Label assignment}

For the label assignment, we rely on the same gadget as the one presented in~\Cref{sec:labelassign}. Using just the fact that the RoS constraint must be satisfied on average across the $T$ rounds (\Cref{def:admissible}, \Cref{item:average-RoS}), implies the following soundness property.
\begin{restatable}{lemma}{learnass}
    \label{lemma:learning-assignment}
    Consider a second-price sequence satisfying the $\beta$-approximate time-average RoS constraints (\Cref{def:admissible}) of the label assignment gadget. Denote $\underline{m}_i^{(t)}$ the second highest multiplier among $(m_{i, \sigma}^{(t)} )_{\sigma \in \Sigma}$. If
    \[
      \Tgood[, i] \defeq \left\{ t\in [T] : \underline{m}_i^{(t)} \leq \lambda \frac{M |\Sigma|}{K} + 1 \right\}
    \]
    for some $\lambda > 1$ and $\beta \leq M$,\footnote{Eventually we will rescale the entries of the label assignment gadget by $\eta$, so $\beta$ needs to be rescaled as well.} then
    $|\Tgood[, i]| \geq T \frac{\lambda - 2}{\lambda}$.
\end{restatable}

(The proofs from this section can be found in~\Cref{appendix:proofs-learning}.) In particular, by a union bound, we obtain the following.

\begin{corollary}
    \label{cor:soundness}
    In the setting of~\Cref{lemma:learning-assignment}, let
    \[
      \Tgood \defeq \left\{ t\in [T] : \underline{m}_i^{(t)} \leq \lambda \frac{M |\Sigma| }{K} + 1 \text{ for all } i \in [d] \right\}.
    \]
    Then $|\Tgood| \geq T ( 1 - \frac{2 d }{\lambda} )$.
\end{corollary}

In other words, we can guarantee soundness \emph{for most rounds}. It is worth pointing out that if $K$ can be set large enough as a function of $T$, satisfying the average RoS constraint implies a bound on the second highest multiplier \emph{for all} rounds, not just on average; however, setting $K$ as a function of $T$ is somewhat restrictive, so we refrain from doing so.

As before, our soundness property does not require having the highest multiplier in each round be close to $M$. \Cref{lemma:learning-assignment} makes allowance for i) setting none of the labels, and ii) half-activating a single label, in the sense that its multiplier will be much larger than $1$, but still far from $M$. Our analysis accounts for those cases because they lead to strictly inferior revenue or welfare.

Unlike the reduction in~\Cref{sec:inapprox-stable}, here we do not have to implement an $\textsf{AND}$ gadget because the starting class of CSPs is simpler. In what follows, we choose $K = \lambda M|\Sigma| /\epsilon$, so that for each $t \in \Tgood$ and $i \in [d]$, $\underline{m}_i^{(t)} \leq 1 + \epsilon$. After the label assignment (\Cref{lemma:learning-assignment}), we proceed directly to the \emph{clause autobidder} interaction, which mirrors~\Cref{tab:edgeautobidder}. Specifically, the clause autobidder $C \in \mathcal{C}$ competes with the autobidders active in that clause, $\{(i,\sigma) \in C\}$, for a single \emph{clause item}. The clause autobidder values that item $1 + \epsilon$, whereas all other autobidders value that item $1/M$. As a result, combining with~\Cref{theorem:Feige}, we are able to show the following hardness result.

\begin{restatable}{theorem}{revenuelearning}
    For any constant $\epsilon > 0$ and a sufficiently small $\beta > 0$, computing a second-price sequence of multipliers that satisfies the $\beta$-approximate time-average RoS constraints  (\Cref{def:admissible}) whose revenue approximates the optimal one within a factor of $e/(e-1) - \epsilon$ is \NP-hard.
\end{restatable}

\subsection{Welfare inapproximability for responsive learning sequences}

To argue about welfare, we introduce a strengthening of~\Cref{def:admissible}. First, it will be convenient to introduce a weak form of undomination, formally defined below.

\begin{definition}
    \label{def:undomination}
    Fix an autobidder $i \in [n]$. We say that the interval $[1, m_i]$ is \emph{pacing dominated} if there exists a multiplier $m_i' > m_i$ such that for any way of resolving ties and all other autobidders' multipliers in $[1, M]^{n-1}$, an item obtained by $i$ through $m_i'$ would also be obtained through $1$. 
\end{definition}
An autobidder has no reason to select a multiplier from a pacing dominated interval. Furthermore, verifying whether a given interval $[1, m_i]$ is pacing dominated by a multiplier $m_i'$ is straightforward using monotonicity: one can first set all other autobidder multipliers' maximally to $M$ and inspect the set of items $i$ can win under a multiplier of $1$. This is to be contrasted with the set of items that can be won under $m_i'$ when all other autobidders bid minimally to 1. A preprocessing step can guarantee that each autobidder only selects pacing undominated multipliers per~\Cref{def:undomination}.

We are now ready to introduce the notion of a \emph{responsive learning sequence}. In what follows, we use the shorthand notation $r_i^{(t)} \defeq \sum_{\tau=1}^t \sum_{j=1}^k x_{i j}^{(\tau)} v_{i j} / \sum_{\tau=1}^t \sum_{j=1}^k x_{i j}^{(\tau)} p_j^{(\tau)}$.

\begin{definition}[Responsive learning sequence]
    \label{def:learning}
    The sequence $(\vm^{(1)}, \vx^{(1)}), \dots, (\vm^{(T)}, \vx^{(T)})$ is an $(\alpha, \beta, \mu)$-\emph{approximate responsive learning sequence} if the following properties hold:
    \begin{enumerate}
        \item It is a second-price sequence satisfying the $\beta$-approximate time-average RoS constraints (\Cref{def:admissible}).
        \item If $ \msafe[,i] = (1 - \mu) \sup_{m_i} m_i$, where the supremum is over all pacing dominated intervals $[1, m_i]$, each autobidder $i \in [n]$ satisfies $m_i^{(t)} \geq \msafe[, i]$ at any time $t \in [T]$ (\Cref{def:undomination}). \label{item:undomination}
        \item There exists a constant $c > 0$ such that for any $s > 0$ the following holds. For each autobidder $i \in [n]$ and times $t_1, t_2 \in [T]$ with the properties\label{item:response}
        \begin{enumerate}
            \item $t_2 - t_1 > \alpha T$ and
            \item $r_i^{(\tau)}  \geq 1+s$ for all $\tau \in [t_1, t_2-1]$,
        \end{enumerate}
        it must be that
        \begin{align}
          \frac{1}{t_2 - t_1} \sum_{\tau={t_1+1}}^{t_2} m_i^{(\tau)} \geq \min\{M, (1 +c s) \msafe[,i] \}. \label{eq:reaction}
        \end{align}
    \end{enumerate}
\end{definition}
The third condition forces the multipliers to respond to a large cumulative surplus, but in a very weak sense: \eqref{eq:reaction} triggers only for long intervals ($> \alpha T$) in which the autobidder had a large surplus throughout. In that case, \eqref{eq:reaction} imposes that the average multiplier in the interval is larger than the safe multiplier (unless the safe multiplier was already approximately maximal). We now describe a class of update rules that produce a responsive learning sequence per~\Cref{def:learning}.

\begin{restatable}{lemma}{psirule}
    \label{lemma:psi}
    Consider the update rule $m_i^{(t+1)} = \psi_i(r_i^{(t)})$ for a nondecreasing function $\psi : (0, \infty) \to [1, M]$ satisfying the following for all $i \in [n]$:
    \begin{itemize}
        \item $\psi_i(s') = \msafe[,i]$ for any $s' < 1$,  and
        \item $\psi_i(1+s) \geq \min\{M, (1 + c s) \msafe \}$ for any $s > 0$, where $c > 0$ is some constant.
    \end{itemize}
    If each autobidder employs this update rule, \Cref{def:learning} is satisfied with $\alpha = 0$, $\beta = O_T(1/T)$, and $\mu$ per the definition of $\msafe[,i]$ in~\Cref{item:undomination}.
\end{restatable}
A simple function $\psi_i$ that meets these preconditions is
\begin{equation}
    \label{eq:step}
    \psi_i(r) \defeq
    \begin{cases}
        \msafe[,i] & \text{if } r < 1, \\
        M & \text{otherwise}.
    \end{cases}
\end{equation}
Two continuous alternatives to the above function are defined below.

\noindent
\begin{minipage}{.48\textwidth}
  \begin{equation}
    \label{eq:poly}
    \psi_i(r) \defeq
    \begin{cases}
        \msafe[,i] & \text{if } r<1,\\
        \min(M, \msafe[,i] r^d) & \text{if } r \geq 1,
    \end{cases}
  \end{equation}
\end{minipage}%
\hfill
\begin{minipage}{.50\textwidth}
  \begin{equation}
    \label{eq:exp}
    \psi_i(r) \defeq
    \begin{cases}
        \msafe[,i] & \text{if } r<1,\\
        \min(M, \msafe[,i] e^{\eta (r-1)} ) & \text{if } r \geq 1.
    \end{cases}
  \end{equation}
\end{minipage}
\par\vspace{\belowdisplayskip}\noindent
Here, $d \in \mathbb{N}$ is the degree of the polynomial, while $\eta > 0$ can be thought of as the learning rate. The update induced by~\eqref{eq:exp} is closer to algorithms such as multiplicative weights update. \Cref{claim:psi} shows that all these update rules satisfy the preconditions of~\Cref{lemma:psi}. 

In this context, similarly to~\Cref{theorem:welfare-hard}, we incorporate an \emph{incumbent clause autobidder} for each clause whose interaction with the clause autobidder is defined as in~\Cref{tab:inc-edge}. By carefully selecting the parameters (\Cref{appendix:proofs-learning}), we are able to show the following hardness result.

\begin{restatable}{theorem}{welfarelearning}
    For any constant $\epsilon > 0$ and sufficiently small $\alpha > 0$, $\beta > 0$, and $\mu > 0$, computing an $(\alpha, \beta, \mu)$-approximate responsive learning sequence per~\Cref{def:admissible} whose welfare approximates the optimal one within a factor of $2e/(2e-1) - \epsilon$ is \NP-hard.
\end{restatable}

The key idea in the proof is that so long as the RoS constraint of each clause autobidder is close to binding, soundness is preserved. In the contrary case, if the clause autobidder has collected a large surplus over a long interval, \Cref{item:response} forces its multiplier to be high, so much so that it still inefficiently captures a large fraction of the item from the incumbent clause autobidder.
\section{Conclusions and future research}

In this paper, we established tight inapproximability results for welfare maximization in autobidding systems: while any autobidding equilibrium is guaranteed to attain an approximation factor of $2$ relative to the optimal liquid welfare, we showed that deciding whether there exists one that is even marginally better is \NP-hard. For revenue, we showed a stronger (sub)logarithmic inapproximability result. We further extended our scope to settings where the auctioneer possesses prior information concerning the underlying valuation profile, as well as to the dynamic setting where autobidders employ learning algorithms---which is prevalent in practice.

Our results suggest several interesting avenues for future research. Our lower bounds for learning dynamics leave a gap relative to existing upper bounds. As discussed, this reflects a clear tradeoff: the more permissive the class of learning algorithms one considers, the more challenging it becomes to establish the soundness of the reduction. Closing this gap is a natural next step. Moreover, our focus was on the second-price auction format. Extending our analysis to other formats, most notably to hybrids between first- and second-price auctions, is a promising direction for the future.


\bibliography{refs}

\clearpage
\appendix

\section{Omitted proofs}
\label{appendix:proofs}

This section contains the proofs omitted from the main body. 

\subsection{Proofs from Section~\ref{sec:prels}}
\label{appendix:proofs2}

For completeness, we begin with the simple proof of the fact that the price of anarchy (PoA) of autobidding equilibria is at most $2$.

\poatwo*

\begin{proof}
    Let $(\vec{m}, \vx)$ be an autobidding equilibrium. We partition the set of items in two sets, $S \subseteq [k]$ and $S' = [k] \setminus S$, so that $S$ contains all items for which an autobidder who values that item the most exclusively won that item. Let $j \in S'$. It then follows that $p_j \geq v_{i^*(j) j}$, where $i^*(j)$ is any autobidder who values item $j$ the most; this holds because the winner of that item must pay at least $m_{i^*} v_{i^* j} \geq v_{i^* j}$ (by~\Cref{def:autobidequil}). As a result, we have
    \begin{equation*}
        \welfare(\vx) - \opt \geq - \sum_{j \in S'} v_{i^* j} \geq - \sum_{j \in S'} p_j \geq - \sum_{i=1}^n \sum_{j=1}^k x_{i j} p_j \geq - \welfare(\vx),
    \end{equation*}
    where the last inequality follows from the RoS constraints. Rearranging yields the claim.
\end{proof}

\subsection{Proofs from Section~\ref{sec:inapprox-stable}}
\label{appendix:proofs-main}

We continue with the proofs from~\Cref{sec:inapprox-stable}.

\labass*

\begin{proof}
    Consider any autobidding equilibrium $(\vec{m}, \vec{x})$ of this instance. We divide the analysis into two cases. 
    
    First, suppose $\max_{\sigma \in \Sigma} m_{u,\sigma} < M$. It then follows that for each $j = 1, \dots, |\Sigma|$, the $j$th item is won exclusively by the $j$th autobidder. Moreover, if there is some autobidder whose multiplier is (strictly) smaller than the highest multiplier, that bidder's RoS constraint is not binding. This holds because that autobidder obtains a single item valued $M$ but pays less than $M$ for that item since $\max_{\sigma \in \Sigma} m_{u,\sigma} < M$. However, maximal pacing (\Cref{def:autobidequil}, \Cref{item:max-bid}) would then imply that the autobidder should have a multiplier equal to $M$, which is a contradiction since $\max_{\sigma \in \Sigma} m_{u,\sigma} < M$. We conclude that in the event where $\max_{\sigma \in \Sigma} m_{u,\sigma} < M$, it must be that all autobidders' multipliers are equal. In particular, given that the RoS constraints are binding, we have
    \begin{equation*}
    \eta( M |\Sigma| + K) = \eta (|\Sigma| m_u + K m_u) \implies m_u = \frac{M |\Sigma| + K }{ |\Sigma| + K} \approx 1
    \end{equation*}
    when $K \gg 1$. The first term in the left-hand side above is the total value obtained by the autobidders and the second term is the total payment. This establishes the claim.

    Otherwise, let $\max_{\sigma \in \Sigma} m_{u,\sigma} = M$. By the RoS constraint of the autobidder who obtains item $|\Sigma| + 1$, it follows that the second highest multiplier must be at most $1 + (M + |\Sigma|)/K \approx 1$ when $K \gg 1$. This holds because the total value that can be obtained by that autobidder is at most $M + |\Sigma| + K$ while the payment is at least the second highest multiplier multiplied by $K$.

    What remains to show is that, for any $\sigma \in \Sigma$, the set of multipliers $(m_{u,\sigma'})_{\sigma' \in \Sigma}$ in which $m_{u,\sigma} = M$ and $m_{u,\sigma'} = 1$ for all $\sigma' \neq \sigma$ is part of an autobidding equilibrium $(\vec{m}, \vec{x})$ for some allocation vector $\vx \in \R^{|\Sigma| \times (|\Sigma| + 1)}$. Indeed, any autobidder $\sigma' \neq \sigma$ can only obtain a positive fraction of item $\sigma'$, but that autobidder pays $M$ for that item since $m_{u,\sigma} = 1$. This means that the RoS constraint of $\sigma'$ is binding. Furthermore, $\sigma$ also satisfies maximal pacing since $m_{u,\sigma} = M$.
\end{proof}

To employ the label assignment gadget going forward, our reduction makes use of a simple fact formally justified below.

\consext*

\begin{proof}
    All items in $S$ are obtained exclusively by autobidders in $N$. Given that each autobidder $i \in N$ satisfies the RoS constraint in $I'$, we have $\sum_{j \in S'} x_{i j} v_{i j} \geq \sum_{j \in S'} x_{i j} p_{j}$. Since $x_{i j} = 0$ for all $j \in S' \setminus S$, it follows that $\sum_{j \in S} x_{i j} v_{i j} \geq \sum_{j \in S} x_{i j} p_{j}$, so the RoS constraint is satisfied. Similarly, if $m_i \neq M$, maximal pacing in $I'$ enforces $\sum_{j \in S'} x_{i j} v_{i j} = \sum_{j \in S'} x_{i j} p_{j}$, which in turn implies $\sum_{j \in S} x_{i j} v_{i j} = \sum_{j \in S} x_{i j} p_{j}$, so the RoS constraint is binding.
\end{proof}

With this notion of conservative extension at hand, we proceed by establishing the correctness of the \textsf{NAND} gadget.

\nandgate*

\begin{proof}
    We consider an instance with three autobidders---the two input autobidders and an output autobidder---and four items, as shown in~\Cref{tab:nand}; we will carry out the analysis under a scaling factor $\eta = 1$ since that does not affect the autobidding equilibrium. The output autobidder competes with the first input autobidder on a single item. The output autobidder values that item $1/2 + \epsilon$, whereas the input autobidder values it $1/(2M)$. The output autobidder also competes with the second input autobidder on a second item valued $1/2 + \epsilon$ and $1/(2M)$, respectively. Given that the multipliers are assumed to be in $[1, M]$, this guarantees that the input autobidder always obtains both items. Finally, there is a third item with a reserve price of $1 + 3 \epsilon$ and value $1$ for the output autobidder and a fourth item with reserve price $0.5$ and value $1/(2M)$, again for the output autobidder.

    We divide the analysis into two cases. First, let $m = M$ and $m' = M$, where $m$ is the multiplier of the first input autobidder and $m'$ is the multiplier of the second input autobidder. If the multiplier of the output autobidder $\widebar{m}$ was higher than $1 + 3 \epsilon$, its RoS constraint would be violated. This holds because it would obtain the first three items for a total price of $2 + 3 \epsilon$, and these are valued collectively at $2 + 2 \epsilon$ from the output autobidder.
    
    In the contrary case, suppose $m \in [1, 1 + \epsilon]$ or $m' \in [1, 1 + \epsilon]$. If $\widebar{m} < M$, the RoS constraint of the output autobidder is not binding, thereby violating maximal bidding.
\end{proof}

We next establish correctness of the \textsf{NOT} gadget.

\negation*

\begin{proof}
    We construct an instance with two autobidders---an input autobidder and an output autobidder---and three items, as shown in~\Cref{tab:neg}. We can again assume that the scaling factor satisfies $\eta = 1$. The input autobidder is competing with the output autobidder on a single item. The input autobidder values that item $1/M$ while the output autobidder values it $1 + \epsilon$. This guarantees that the item is always obtained by the output autobidder. Moreover, there is a second item with reserve price $1 + 2 \epsilon$ and value $1$ for the output autobidder and a third item with reserve price $1$ and value $1/M$, again for the output autobidder.

    Now, if $m = M$, the RoS constraint of the output autobidder implies $\widebar{m} \leq 1 + 2 \epsilon$. This holds because the output autobidder obtains the first item at a price of $1$, and so a multiplier higher than $1 + 2 \epsilon$ would also procure the second item for that autobidder; this is a violation of the RoS constraint since the first two items are valued at $2 + \epsilon$ and the total price paid is $2 + 2 \epsilon$. In the contrary case, let $m \in [1, 1 + \epsilon]$. Any multiplier $\widebar{m}$ smaller than $M$ would not bind the RoS constraint of the output autobidder, violating maximal pacing (\Cref{def:autobidequil}, \Cref{item:max-bid}).
\end{proof}

Combining these gadgets, we first argue about approximating the optimal revenue.

\revenuehard*

\begin{proof}
    If there is a labeling $\ell : V \to \Sigma$ that satisfies all the edge constraints, \eqref{eq:revenue-lb} implies that $\revenue \geq |E|$ (by completeness of the reduction). On the other hand, if any labeling $\ell$ satisfies at most a $\delta$ fraction of the edge constraints, \eqref{eq:revenue-ub} implies that $\revenue \leq 3 \delta |E|$ (soundness of the reduction) when $M \geq (1 + \epsilon)/\delta$ and
\begin{equation}
    \label{eq:eta-bound}
    \eta \leq \frac{\delta |E|}{|V| ( |\Sigma| M + K ) + 5 |E| |\Sigma|}.
\end{equation}
In more detail, we define $z_{u, \sigma} = \mathbbm{1} \{m_{u,\sigma} = M \} $ for each $u \in V$ and $\sigma \in \Sigma$. By definition, $z_{u, \sigma} \in \{0, 1\}$. Furthermore, by~\Cref{lemma:assignment}, $\sum_{\sigma \in \Sigma } z_{u, \sigma} \leq 1$ in any autobidding equilibrium. If $z_{e, \sigma} = \mathbbm{1} \{ m_{e, \sigma} = M \}$, \eqref{eq:edge-sigma} implies that $z_{e, \sigma} = 1$ if and only if $z_{u, \sigma} = 1$ and $z_{v, \Pi_e(\sigma)} = 1$; in other words, $z_{e, \sigma} = z_{u, \sigma} \land z_{v,\Pi_e(\sigma)}$. We conclude that, in equilibrium,
\[
    \sum_{e \in E} \bigvee_{\sigma \in \Sigma} ( m_{e, \sigma} = M) = \sum_{e \in E} \bigvee_{\sigma \in \Sigma} z_{e, \sigma} = \sum_{e \in E} \bigvee_{\sigma \in \Sigma} (z_{u, \sigma} \land z_{v,\Pi_e(\sigma)})
\]
subject to $\sum_{\sigma \in \Sigma } z_{u, \sigma} \leq 1$ for all $u \in V$.

In conclusion, we can set the parameters in the construction as follows.
\begin{itemize}
    \item $\epsilon > 0$ and $\delta > 0$ are any sufficiently small absolute constants;
    \item $\Sigma = \Sigma(\delta)$ per~\Cref{theorem:labelcover}.
    \item $M = (1 + \epsilon)/\delta$;
    \item $K = 6 M |\Sigma|/\epsilon$; and
    \item $\eta > 0$ per~\eqref{eq:eta-bound}.
\end{itemize}
\end{proof}

To extend our reduction to welfare, we first prove a crucial lemma concerning the interaction of the edge autobidder with the incumbent edge autobidder.

\stealing*

\begin{proof}
    Let us assume first that $m_{e, \sigma} = M$ for some label $\sigma \in \Sigma$. Suppose further that the edge autobidder obtains a nonzero fraction of the item in question. This would imply that the payment for that item would be at least $1$. For the edge autobidder to satisfy its RoS constraint, it must be the case that it is allocated at most an $2 \epsilon$ fraction of the item. Indeed, if $x$ is the fraction, the RoS constraint yields 
    \[
      1+ \epsilon + \frac{x(1+\epsilon)}{M} \geq 1 + x  \implies x \leq 2 \epsilon,
    \]
    where we used the fact that $m_{e, \sigma} = M$ for some label $\sigma \in \Sigma$.
    
    We now analyze the contrary case, that is, $m_{e, \sigma} \leq 1 + \epsilon$ for all $\sigma \in \Sigma$. We first claim that the edge autobidder bids at least as high as the incumbent edge autobidder for the item. If not, the incumbent edge autobidder would exclusively win the item, which would in turn mean that the edge autobidder would have to bid maximally since its RoS constraint is not binding. But this is only possible if the payment for the item is strictly more than $1$, violating the RoS constraint of the incumbent edge autobidder. We have thus established that the edge autobidder bids at least as high as the incumbent edge autobidder. If the edge autobidder bids strictly higher than the incumbent, it exclusively wins the item. So it remains to analyze the case where both bid the same amount, and that amount is equal to $1$. In that case, since the edge autobidder is not bidding maximally, it follows that its RoS constraint is binding. If the edge autobidder is winning an $x$ fraction of the item, we have
    \[
        1 + \epsilon + x \frac{1 + \epsilon}{M} = \frac{1 + \epsilon}{M} + x,
    \]
    which is impossible. In other words, in an autobidding equilibrium, the edge autobidder must be bidding higher than the incumbent edge autobidder.
\end{proof}

We continue with the proofs from~\Cref{sec:MLadvise}.

\gammarevenue*

\begin{proof}
    As before, if there is a labeling $\ell : V \to \Sigma$ that satisfies all the edge constraints, the completeness of the reduction guarantees that $\revenue \geq |E|$. We now analyze the contrary case, where any labeling $\ell$ satisfies at most a $\delta$ fraction of the edge constraints. It then follows that for any autobidding equilibrium $(\vm, \vx)$,
    \[
        \revenue(\vm) \leq \eta (|V| ( |\Sigma| M + K ) + 5 |E| |\Sigma|) + \gamma |E| + (1-\gamma) \sum_{e \in E} \bigvee_{\sigma \in \Sigma} (m_{e, \sigma} = M) + \frac{(1+\epsilon)|E|}{M}.
    \]
    In particular, the price paid by the edge autobidder $e \in E$ is either at most $\gamma$ (assuming $\gamma \geq \frac{1+\epsilon}{M}$) or $1$, and the latter only happens if there exists $\sigma \in \Sigma$ such that $m_{e, \sigma} = M$. So, if at most a $\delta$ fraction of the edge constraints are satisfied,
    \[
        \revenue(\vm) \leq \eta (|V| ( |\Sigma| M + K ) + 5 |E| |\Sigma|) + \gamma |E| + (1-\gamma) \delta |E| + \frac{(1+\epsilon)|E|}{M}.
    \]
    We set the parameters as follows.
\begin{itemize}
    \item $\epsilon > 0$ and $\delta > 0$ are any sufficiently small absolute constants;
    \item $\Sigma = \Sigma(\delta)$ per~\Cref{theorem:labelcover}.
    \item $M = (1 + \epsilon)/\delta$;
    \item $K = 6 M |\Sigma|/\epsilon$; and
    \item $\eta > 0$ per~\eqref{eq:eta-bound}.
\end{itemize}
    As a result,
    \[
        \revenue(\vm) \leq 3 \delta |E| + \gamma (1 - \delta) |E| \leq 3 \delta |E| + \gamma |E|,
    \]
    and the claim follows.
\end{proof}

We next adapt~\Cref{lemma:stealing} in the presence of reserve prices set to $\gamma$-approximate signals.

\refinedstealing*

\begin{proof}
    The first part of the analysis of~\Cref{lemma:stealing}, when $m_e^{(\sigma)}$ for some $\sigma \in \Sigma$, remains the same. So we analyze the case where $m_e^{(\sigma)} \leq 1 + \epsilon$ for all $\sigma \in \Sigma$. Specifically, it suffices to analyze the case where both autobidders bid $1$ for the item. Since the edge autobidder is not bidding maximally, it follows that its RoS constraint must be binding. In particular, if the edge autobidder obtains an $x$ fraction of that item, we have
    \begin{equation*}
        1 + \epsilon - \gamma (1 + \epsilon) + x \frac{1 + \epsilon}{M} - x = 0,
    \end{equation*}
    which implies that $x \geq 1 - \gamma$, as claimed.
\end{proof}

\gammawelfare*

\begin{proof}
    As before, \Cref{lemma:refined-stealing} shows that there exists an autobidding equilibrium $(\vm, \vx)$ such that
    \[
        \welfare(\vx) \geq (1 +\epsilon) |E| + ( 1 - 2 \epsilon) \sum_{e \in E}  \bigvee_{\sigma \in \Sigma} ( m_{e, \sigma} = M ).
    \]
    Turning to the soundness of the reduction, we can again use~\Cref{lemma:refined-stealing} to obtain
    \[
        \welfare \leq \eta \left( |V| ( |\Sigma| M + K ) + 5 |E| |\Sigma| \right) + (1+\epsilon) |E| + \sum_{e \in E}  \bigvee_{\sigma \in \Sigma} ( m_{e, \sigma} = M ) + \gamma |E| + \frac{|E| ( 1 + \epsilon) }{M}.
    \]
    As a result, if there exists a labeling that satisfies all the edge constraints, there exists an autobidding equilibrium $(\vm, \vx)$ such that
    \[
        \welfare(\vx) \geq (2 - \epsilon) |E|.
    \]
    Setting the parameters as before, if at most $\delta$ of the edge constraints can be satisfied, then any autobidding equilibrium $(\vm, \vx)$ satisfies
    \[
        \welfare(\vx) \leq 2 \delta |E| + \gamma |E| + (1+\epsilon)|E|,
    \]
    and the claim follows.
\end{proof}

\subsection{Proofs from Section~\ref{sec:learning}}
\label{appendix:proofs-learning}

Finally, we conclude with the proofs from~\Cref{sec:learning}, starting from~\Cref{lemma:learning-assignment}.

\learnass*

\begin{proof}
    Let $\underline{m}_i^{(t)}$ be the second highest multiplier in each round $t \in [T]$. The total value obtained by the autobidders in each round is at most $M|\Sigma| + K$, while the total payment in each round is at least $K \underline{m}_i^{(t)}$. Combining the RoS constraints of all autobidders,
    \begin{equation}
        \label{eq:av-RoS}
        \frac{1}{T} \sum_{t=1}^T \underline{m}_i^{(t)} \leq \frac{M|\Sigma|}{K} + 1 + \frac{\beta |\Sigma|}{K} \leq 1 + 2\frac{M |\Sigma|}{K},
    \end{equation}
    where we used the bound on $\beta$. Using the definition of $\Tgood[, i]$ and combining with~\eqref{eq:av-RoS},
    \begin{align*}
        \left( \lambda \frac{M ( |\Sigma| - 1 ) }{K} + 1 \right) \left( 1 - \frac{|\Tgood[, i]|}{T} \right) + \frac{|\Tgood[, i]|}{T} &\leq \frac{1}{T} \sum_{t \in [T] \setminus \Tgood[, i]} \underline{m}_i^{(t)} + \frac{1}{T} \sum_{t \in \Tgood[, i]} \underline{m}_i^{(t)} \\
        &\leq 2 \frac{M |\Sigma|}{K} + 1. 
    \end{align*}
    This implies
    \[
        |\Tgood[, i]| \geq T \frac{\lambda-2}{\lambda} .
    \]
\end{proof}

We use this soundness property to argue about revenue under any second-price sequence satisfying the time-average RoS constraints.

\revenuelearning*

\begin{proof}
The average revenue $\frac{1}{T} \sum_{t=1}^T \revenue(\vm^{(t)})$ can be upper bounded by
\begin{equation}
    \label{eq:avg-revenue-ub}
    \eta ( d ( |\Sigma| M + K )) + \left( 1 - \frac{|\Tgood|}{T} \right) |\mathcal{C}| + \frac{1}{T} \sum_{t \in \Tgood} \sum_{C \in \mathcal{C}}  \bigvee_{(i, \sigma) \in C} (m_{i, \sigma}^{(t)} \geq 2) + \frac{2|\mathcal{C}|}{M}.
\end{equation}
Let $\opt(\mathcal{C})$ denote the maximum number of clauses that can be satisfied in the CSP instance. By the soundness property of $\Tgood$ (\Cref{cor:soundness}), it follows that for each $t \in \Tgood$,
\[
    \sum_{C \in \mathcal{C}}  \bigvee_{(i, \sigma) \in C} (m_{i, \sigma}^{(t)} \geq 2) \leq \opt(\mathcal{C}).
\]
Combining with the bound on $|\Tgood|$ established in~\Cref{cor:soundness}, we can upper bound~\eqref{eq:avg-revenue-ub} by
\[
    \eta ( d ( |\Sigma| M + K )) + \frac{2d |\mathcal{C}| }{\lambda} + \frac{2 |\mathcal{C}|}{M} + \opt(\mathcal{C}).
\]
We set the parameters as follows.
\begin{itemize}
    \item $\epsilon > 0$ and $\delta > 0$ are sufficiently small constants;
    \item $\Sigma$ is determined by the original CSP;
    \item $M = 2/\delta$;
    \item $\lambda = 2d/\delta$;
    \item $K = \lambda M (|\Sigma| - 1)/\epsilon$;
    \item $\eta = \delta |\mathcal{C}|/(d (|\Sigma| M + K) )$; and
    \item $\beta = \eta M$. (Compared to~\Cref{lemma:learning-assignment}, we have rescaled the values in the label assignment gadget by $\eta$, which is reflected in the value of $\beta$.)
\end{itemize}
This choice guarantees
\[
    \frac{1}{T} \sum_{t=1}^T \revenue(\vm^{(t)}) \leq 3 \delta |\mathcal{C}| + \opt (\mathcal{C})
\]
for any second-price sequence satisfying the time-averate RoS constraints. Moreover, by completeness of the reduction, we also know that there is an autobidding equilibrium $(\vm, \vx)$---which in particular satisfies the preconditions of~\Cref{def:admissible}---such that
\[
    \revenue(\vm) \geq \opt(\mathcal{C}).
\]
Combining with~\Cref{theorem:Feige}, the claim follows.
\end{proof}

\psirule*

\begin{proof}
The fact that each autobidder $i \in [n]$ is $\mu$-pacing undominated follows by definition. We next verify~\Cref{item:response} of~\Cref{def:learning}. Consider any $s > 0$ and an interval $[t_1, t_2]$ satisfying the conditions of~\Cref{item:response}. Since $m_i^{(\tau)} = r_i^{(\tau-1)}$, it follows that $m_i^{(\tau)} \geq \psi(1 + s)$ for all $\tau \in [t_1 + 1, t_2]$ since $\psi$ is nondecreasing. As a result, $m_i^{(\tau)} \geq \min \{M, (1 + c s)\msafe \}$ for all $\tau \in [t_1 + 1, t_2]$, and~\eqref{eq:reaction} follows.

We finally argue that the update rule satisfies the time-average RoS constraint. Let $t' \in [T]$ be the earliest time such that $m_i^{(\tau)} = \msafe$ for all $\tau \geq t'$. If $t' \in \{1, 2\}$ the claim follows trivially, so let us assume $t' \geq 3$. Since $t'$ is the earliest such time, it follows that $m_i^{(t' - 1)} > \msafe$, which in turn implies $r_i^{(t'-2)} \geq 1 \iff \sum_{\tau = 1}^{t' - 2} \sum_{j=1}^k ( x_{i j}^{(\tau)} v_{i j} - x_{i j}^{(\tau)} p_j^{(\tau)} ) \geq 0$. Moreover, we also have $\sum_{j=1}^k ( x_{i j}^{(\tau)} v_{i j} - x_{i j}^{(\tau)} p_j^{(\tau)} ) \geq 0$ for all $\tau \geq t'$. This holds because for all these rounds the multiplier $\msafe$ was selected, which belongs to a pacing dominated interval. This implies that the price of any fraction of an item obtained under $\msafe$ is at most its value for it, since the item would also be won under a multiplier $1$. Combining,
\begin{align*}
    \frac{1}{T} \sum_{t = 1}^{T} \sum_{j=1}^k ( x_{i j}^{(t)} v_{i j} - x_{i j}^{(t)} p_j^{(t)} ) &= \frac{1}{T} \sum_{j=1}^k ( x_{i j}^{(t' - 1)} v_{i j} - x_{i j}^{(t' - 1)} p_j^{(t - 1)} )  +  \frac{1}{T} \sum_{t \neq t'-1 }^{T} \sum_{j=1}^k ( x_{i j}^{(t)} v_{i j} - x_{i j}^{(t)} p_j^{(t)} ) \\
    &\geq \frac{1}{T} \sum_{j=1}^k ( x_{i j}^{(t' - 1)} v_{i j} - x_{i j}^{(t' - 1)} p_j^{(t - 1)} ).
\end{align*}
So, the time-average RoS constraint is satisfied with an $O_T(1/T)$ additive slackness, as claimed.
\end{proof}

\begin{claim}[Update rules that produce responsive learning sequences]
    \label{claim:psi}
    Each of the functions introduced in~\eqref{eq:step}, \eqref{eq:poly}, and~\eqref{eq:exp} satisfies the preconditions of~\Cref{lemma:psi}.
\end{claim}

\begin{proof}
    The step function~\eqref{eq:step} satisfies the conditions of~\Cref{lemma:psi} with constant $c = 1$. The polynomial function~\eqref{eq:poly} satisfies the conditions with constants $c = d$ since $(1+s)^d \geq (1 + s d)$ (Bernoulli's inequality). Finally, the exponential function~\eqref{eq:exp} satisfies the conditions with constants $c = \eta$ since $e^{\eta s} \geq 1 + \eta s$.
\end{proof}

We conclude by extending our reduction to welfare under responsive learning sequences.

\welfarelearning*

\begin{proof}
We observe that~\Cref{item:undomination} of~\Cref{def:learning} implies that the clause autobidder always bids at least $1 - \mu$ for the corresponding incumbent clause item. In particular, $\msafe[,C] = (1 -\mu) \frac{M}{1+\epsilon}$. Let $x_C^{(\tau)}$ be the fraction of the incumbent clause item won by the clause autobidder $C \in \mathcal{C}$ at time $\tau$ and $p_C^{(\tau)}$ be the price paid by the clause autobidder for the corresponding clause item. The total value obtained by the clause autobidder after $t \in [T]$ rounds reads
\begin{equation}
    \label{eq:total-value}
    (1 + \epsilon) t + \frac{1 + \epsilon}{M} \sum_{\tau=1}^{t} x^{(\tau)}_C 
\end{equation}
while the total payment is upper bounded by
\begin{equation}
    \label{eq:totalpay}
    \sum_{\tau=1}^{t} p_C^{(\tau)} + ( 1+\epsilon) \sum_{\tau=1}^{t} x_C^{(\tau)}.
\end{equation}
This holds because the clause autobidder never bids higher than $1+\epsilon$ for the incumbent clause item, so when it wins a positive fraction of that item its price is at most $1+\epsilon$. We will set the parameters of the reduction as follows.
\begin{itemize}
    \item $\epsilon > 0$ is an arbitrarily small constant;
    \item $\mu = \epsilon^2$;
    \item $s > 0$ is such that $(1 + c s) = (1+\epsilon)/(1 - \mu)$ (given that $\mu = \epsilon^2$, $s = \Theta(\epsilon)$);
    \item $M = 1/\epsilon$;
    \item $\alpha = \epsilon$; 
    \item $\lambda = 2d/\epsilon$;
    \item $K = \lambda M (|\Sigma| - 1)/\epsilon$;
    \item $\eta = \epsilon |\mathcal{C}|/(d (|\Sigma| M + K) )$; and
    \item $\beta = \min \{ \epsilon^2, \eta M\}$.
\end{itemize}
Our goal is to prove that
\begin{equation}
    \label{eq:goal}
    \frac{1}{T} \sum_{t=1}^T x_C^{(t)} \geq 1 - \frac{1}{T} \sum_{t=1}^T  p_C^{(t)} - \Theta(\epsilon).
\end{equation}
Suppose that $t_1$ was the earliest time such that $r_C^{(\tau)} \geq 1 +s$ for all $\tau \geq t_1$; we can assume that such a time exists and is different from $1$ (the contrary case can be treated similarly). Since $t_1$ is the earliest such time, $r_C^{(t_1 - 1)} < 1 + s$. As a result, combining~\eqref{eq:total-value} and~\eqref{eq:totalpay},
\begin{equation}
    \label{eq:1+s}
    (1 + \epsilon) (t_1 - 1) + \frac{1 + \epsilon}{M} \sum_{\tau=1}^{t_1 - 1} x^{(\tau)}_C < (1+s) \left( \sum_{\tau=1}^{t_1 - 1} p_C^{(\tau)} + ( 1+\epsilon) \sum_{\tau=1}^{t_1 - 1} x_C^{(\tau)} \right).
\end{equation}
We now analyze the interval $[t_1, T]$. By construction, we have $r_C^{(\tau)} \geq 1 + s$ for all $\tau$ in this interval. If $T - t_1 \leq \alpha T$, then again~\eqref{eq:goal} follows since the interval $[t_1, T]$ has negligible impact on the average. Specifically,
\[
    (1 + \epsilon) + \frac{1 + \epsilon}{M} \frac{1}{T} \sum_{t=1}^{T} x^{(t)}_C < (1+s) \left( \frac{1}{T} \sum_{t=1}^{T} p_C^{(t)} + ( 1+\epsilon) \frac{1}{T} \sum_{t=1}^{T} x_C^{(t)} \right) + 2 \alpha T
\]
since
\[
    (1 + \epsilon) (T - t_1+1) + \frac{1 + \epsilon}{M} \sum_{\tau=t_1}^{T} x^{(\tau)}_C - (1+s) \left( \sum_{\tau=t_1}^{T} p_C^{(\tau)} - ( 1+\epsilon)  \sum_{\tau=t_1}^{T} x_C^{(\tau)} \right) \leq 2 \alpha.
\]
In the contrary case, we can use~\eqref{eq:reaction} to obtain
\[
    \frac{1}{T - t_1} \sum_{\tau = t_1+1}^T m_C^{(\tau)} \geq \min\{M, (1 + c s) \msafe[,C] \} \geq M,
\]
by our choice of $s$, where $m_C^{(\tau)}$ denotes the multiplier of the clause autobidder at time $\tau$. In particular, this means that $m_C^{(\tau)} \geq M$ for all $\tau \in [t_1 + 1, T]$. Now, the incumbent clause autobidder obtains a $1 - x_C^{(t)}$ fraction of the incumbent clause item at each round $t$. The RoS constraint of the incumbent clause autobidder yields
\[
    \mu t_1 + \sum_{\tau=t_1 + 1}^T (1 - x_C^{(t)}) \geq \sum_{\tau = t_1+1}^T (1 - x_C^{(t)}) \left( \frac{1+\epsilon}{M} m_C^{(t)} \right) - \beta T.
\]
Here we used the fact that in the first $t_1$ rounds the incumbent clause autobidder pays at least $1 - \mu$ for the incumbent clause item. Using the fact that $m_C^{(\tau)} \geq M$ for all $\tau \in [t_1 + 1, T]$, we get
\[
    \sum_{\tau = t_1 + 1}^T (1 - x_C^{(t)}) \leq \frac{1}{\epsilon} ( \mu t_1 + \beta T).
\]
By our choice of parameters,
\[
    \sum_{\tau = t_1 + 1}^T (1 - x_C^{(t)}) \leq 2 \epsilon T.
\]
Combining with~\eqref{eq:1+s}, \eqref{eq:goal} follows. Now, if $t \in \Tgood$, \Cref{cor:soundness} implies that $p_C^{(t)} \leq 2/M$ or at most $1$ only if there exists $(i, \sigma) \in C$ such that $m_{i, \sigma}^{(t)} \geq 2$. That is, we have shown that for any clause autobidder $C \in \mathcal{C}$,
\[
    \frac{1}{T} \sum_{t=1}^T x_C^{(t)} \geq 1 - \frac{1}{T} \sum_{t \in \Tgood} \bigvee_{(i, \sigma) \in C} (m_{i, \sigma}^{(t)} \geq 2) - \Theta ( \epsilon ),
\]
where we also used the fact that $\lambda = 2d/\epsilon$ together with~\Cref{cor:soundness} that guarantees $|\Tgood| \geq T \left(1 - \frac{2d}{\lambda} \right)$. As a result,
\begin{align*}
    \frac{1}{T} \sum_{t=1}^T \welfare(\vx^{(t)}) &\leq \eta ( d (|\Sigma| M + K ) ) + (1+\epsilon) |\mathcal{C}| + \frac{1}{T} \sum_{t=1}^T (1-x_C^{(t)}) + \frac{1}{T} \sum_{t=1}^T x_C^{(t)} \frac{1+\epsilon}{M} \\
    &\leq |\mathcal{C}| + \frac{1}{T} \sum_{t \in \Tgood} \sum_{C \in \mathcal{C}}  \bigvee_{(i, \sigma) \in C} (m_{i, \sigma}^{(t)} \geq 2) + |\mathcal{C}| \Theta(\epsilon),
\end{align*}
and the claim follows from~\Cref{theorem:Feige}.
\end{proof}
\section{Simulating reserve prices}
\label{appendix:reserveprices}

\citet{Li24:Vulnerabilities} gave a basic gadget with two auxiliary autobidders and two auxiliary items that simulates the presence of reserve prices (\Cref{tab:reserve}). For completeness, we include the simple proof below.\footnote{If one incorporates these auxiliary autobidders in the reduction of~\Cref{sec:inapprox-stable}, the welfare and revenue changes by an additive factor $O(\eta)$ since they only affect the low-stake gadgets, so all of the consequences are preserved.}

\begin{table}[t]
    \centering
    \caption{The reserve price gadget of~\citet{Li24:Vulnerabilities}. }
    \begin{tabular}{c c c c}
     & Item $1$ & Item $2$ & Target item \\ 
     Autobidder 1 & $0$ & $2r$ & $r$ \\
     Autobidder 2 & $r/2$ & $r$ & $r/2$ \\
     Target autobidder & $0$ & $0$ & $v$
\end{tabular}
    \label{tab:reserve}
\end{table}

\begin{lemma}
    In any autobidding equilibrium of the instance described in~\Cref{tab:reserve}, the multiplier of Autobidder 1 is 1 and the multiplier of Autobidder 2 is 2. This is equivalent to having a reserve price of $r$ attached to the target item.
\end{lemma}

\begin{proof}
    Autobidder 2 always selects a multiplier of at least 2, otherwise its RoS constraint is not binding. Moreover, we claim that there is no autobidding equilibrium in which Autobidder 1 selects a multiplier higher than 1. Indeed, in the contrary case, Autobidder 2 would necessarily select a multiplier strictly higher than 2, otherwise its RoS constraint is again not binding. However, this would imply that the price of the second item is higher than $2r$ and the price of the target item is higher than $r$. By the RoS constraint of Autobidder 1, that autobidder cannot win any positive fraction of either item. As a result, Autobidder 1 exclusively wins the second item, violating its RoS constraint. We conclude that in any autobidding equilibrium Autobidder 1 selects a multiplier 1. This also implies that the multiplier of Autobidder 2 cannot be higher than 2, for otherwise it would violate its RoS constraint by exclusively winning the second item.
\end{proof}
\section{Inapproximability under ML advice}
\label{appendix:MLadvise}

This section contains the simple extensions of~\Cref{lemma:nand,lemma:negation} under approximate signals per~\Cref{def:approx-signal}; \Cref{lemma:assignment} is essentially the same.

\begin{lemma}[Refinement of~\Cref{lemma:nand}]
    Let $m \in [1, 1+\epsilon] \cup \{M\}$ and $m' \in [1, 1 + \epsilon] \cup \{M\}$ be the multipliers of two autobidders for some $\epsilon > 0$. There is a conservative extension that produces an output autobidder whose multiplier $\widebar{m}$ satisfies, in equilibrium,
    \begin{equation*}
        \widebar{m} \in 
        \begin{cases}
            [1, 1 + 3 \epsilon] & \text{if } m = M \text{ and } m' = M, \\
            \{M\} & \text{otherwise}.
        \end{cases}
    \end{equation*}
    This holds even under reserve prices set to some $\gamma$-approximate signals per~\Cref{def:approx-signal}.
\end{lemma}

\begin{proof}
    The construction is the same as in~\Cref{lemma:nand} (\Cref{tab:nand}), except that now the first and second items have a reserve price of $\gamma (1/2 + \epsilon)$; for the third and fourth items the existing reserve prices subsume the reserve prices coming from the approximate signals. (As before, we are working under the assumption that $\eta = 1$.)

    If $m = M$ and $m' = M$, where $m$ is the multiplier of the first input autobidder and $m'$ is the multiplier of the second input autobidder, it follows that the multiplier of the output autobidder satisfies $\widebar{m} \leq 1 + 3\epsilon $; otherwise, its RoS constraint would be violated.
    
    In the contrary case, suppose $m \in [1, 1 + \epsilon]$ or $m' \in [1, 1 + \epsilon]$. If both are in $[1, 1 + \epsilon]$, when $M$ is large enough the output autobidder pays $\gamma (1/2 + \epsilon)$ for each of the first two items, for a total of $\gamma + 2 \epsilon \gamma$. Accounting for the third item as well, the total payment is $1 + \gamma + 2 \epsilon \gamma + 3 \epsilon $, whereas the total value is $2 + 2 \epsilon$. As a result, the output autobidder will select a multiplier of $M$, otherwise maximal pacing is violated. Similarly, if only one of the input multipliers is in $[1, 1 + \epsilon]$, the total payment for the first three items is $1/2 + \gamma (1/2 + \epsilon) + 1 + 3 \epsilon $, whereas the total value is again $2 + 2 \epsilon$. When $\epsilon$ is small enough, $1/2 + \gamma (1/2 + \epsilon) + 1 + 3 \epsilon < 2 + 2 \epsilon$, so the output autobidder will bid maximally.
\end{proof}

\begin{lemma}[Refinement of~\Cref{lemma:negation}]
    Let $m \in [1, 1 + \epsilon] \cup \{M\}$ be the multiplier of an input autobidder. There is a conservative extension that produces an output autobidder whose multiplier satisfies, in equilibrium,
    \begin{equation*}
        \widebar{m} \in 
        \begin{cases}
            [1, 1 + 2 \epsilon] & \text{if } $m = M$,\\
            \{M\} & \text{otherwise}.
        \end{cases}
    \end{equation*}
    This holds even under reserve prices set to some $\gamma$-approximate signals per~\Cref{def:approx-signal}.
\end{lemma}

\begin{proof}
    We consider the same instance as in~\Cref{lemma:negation} (\Cref{tab:neg}), except that the first item has a reserve price of $\gamma (1 + \epsilon)$; for the second and third items the existing reserve prices subsume the reserve price stemming from the approximate signals.

    As before, if $m = M$, the RoS constraint of the output autobidder implies $\widebar{m} \leq 1 + 2 \epsilon$. In the contrary case, let $m \in [1, 1 + \epsilon]$. The total price for the fist two items is $\gamma(1 + \epsilon) + 1 + 2 \epsilon$, whereas the total value is $2 + \epsilon$. When $\epsilon$ is small enough, $\gamma(1 + \epsilon) + 1 + 2 \epsilon < 2 + \epsilon$, so maximal pacing forces the output autobidder to select $\widebar{m} = M$.
\end{proof}

\end{document}